\documentclass[10pt,a4paper]{article}
\usepackage[latin1]{inputenc}
\usepackage{setspace}
\usepackage{amsmath}
\usepackage{amsfonts}
\usepackage{amssymb}
\usepackage{amsthm}
\usepackage{bm}
\usepackage{graphicx}
\usepackage{enumitem}
\graphicspath{ {./figures/} }

\usepackage[left=1in, right=1in]{geometry}
\usepackage[title]{appendix}
\usepackage{authblk}
\usepackage{xcolor}

\usepackage{hyperref}

\usepackage[mathscr]{euscript}
\usepackage{bbm}

\newtheorem{theorem}{Theorem}
\newtheorem{lemma}{Lemma}

\theoremstyle{definition}
\newtheorem{conditionA}{Condition}
\newtheorem{conditionB}{Condition}
\newtheorem{conditionC}{Condition}
\newtheorem{conditionD}{Condition}
\newtheorem{conditionE}{Condition}

\newtheorem{remark}{Remark}

\newenvironment{condition2}[1]
{\renewcommand{\theconditionC}{\ref{#1}$^*$}%
    \addtocounter{conditionC}{-1}%
    \begin{conditionC}}
    {\end{conditionC}}

\newenvironment{conditionsufficient}[1]
{\renewcommand{\theconditionC}{\ref{#1}s}%
    \addtocounter{conditionC}{-1}%
    \begin{conditionC}}
    {\end{conditionC}}

\newenvironment{conditionsufficient2}[1]
{\renewcommand{\theconditionC}{\ref{#1}$^*$s}%
    \addtocounter{conditionC}{-1}%
    \begin{conditionC}}
    {\end{conditionC}}

\newtheorem{example}{Example}

\newcommand{\ind}{\mathbbm{1}}
\newcommand{\real}{\mathbb{R}}

\newcommand{\expect}{\mathbb{E}}
\newcommand{\Var}{\textnormal{Var}}

\newcommand{\expit}{\text{expit}}

\newcommand{\id}{\mathcal{I}}
\newcommand{\vnorm}{{\mathrm{v}}}
\DeclareMathOperator*{\argmin}{argmin}

\newcommand{\bfitem}[1]{\noindent\textbf{#1}}

\usepackage[width=.9\textwidth]{caption}

\usepackage[square,numbers]{natbib}
\title{Universal sieve-based strategies for efficient estimation\\using machine learning tools}
\author[*]{Hongxiang Qiu}
\author[$\dagger$]{Alex Luedtke}
\author[*]{Marco Carone}
\affil[*]{Department of Biostatistics, University of Washington}
\affil[$\dagger$]{Department of Statistics, University of Washington}

\date{}

\begin{document}
\maketitle

\begin{abstract}
    Suppose that we wish to estimate a finite-dimensional summary of one or more function-valued features of an underlying data-generating mechanism under a nonparametric model. One approach to estimation is by plugging in flexible estimates of these features. Unfortunately, in general, such estimators may not be asymptotically efficient, which often makes these estimators difficult to use as a basis for inference. Though there are several existing methods to construct asymptotically efficient plug-in estimators, each such method either can only be derived using knowledge of efficiency theory or is only valid under stringent smoothness assumptions. Among existing methods, sieve estimators stand out as particularly convenient because efficiency theory is not required in their construction, their tuning parameters can be selected data adaptively, and they are universal in the sense that the same fits lead to efficient plug-in estimators for a rich class of estimands. Inspired by these desirable properties, we propose two novel universal approaches for estimating function-valued features that can be analyzed using sieve estimation theory. Compared to traditional sieve estimators, these approaches are valid under more general conditions on the smoothness of the function-valued features by utilizing flexible estimates that can be obtained, for example, using machine learning.
\end{abstract}

\section{Introduction}

\subsection{Motivation}

A common statistical problem consists of using available data in order to learn about a summary of the underlying data-generating mechanism. In many cases, this summary involves function-valued features of the distribution that cannot be estimated at a parametric rate under a nonparametric model --- for example, a regression function or the density function of the distribution. Examples of useful summaries involving such features include average treatment effects \citep{Rubin1974}, average derivatives \citep{Hardle1989}, moments of the conditional mean function \citep{Shen1997}, variable importance measures \citep{Williamson2017} and treatment effect heterogeneity measures \citep{Levy2018}. For ease of implementation and interpretation, in traditional approaches to estimation, these features have typically been restricted to have simple forms encoded by parametric or restrictive semiparametric models. However, when these models are misspecified, both the interpretation and validity of subsequent inferences can be compromised. To circumvent this difficulty, investigators have increasingly relied on machine learning (ML) methods to flexibly estimate these function-valued features.

Once estimates of the function-valued features are obtained, it is natural to consider plug-in estimators of the summary of interest. However, in general, such estimators are not root-$n$-consistent and asymptotically normal, and hence not asymptotically efficient (referred to as \emph{efficient} henceforth). Lacking this property is problematic since it often forms the basis for constructing valid confidence intervals and hypothesis tests \citep{Bickel2003,Newey2004}. When the function-valued features are estimated by ML methods, in order for the plug-in estimator to be CAN, the ML methods must not only estimate the involved function-valued features well, but must also satisfy a small-bias property with respect to the summary of interest \citep{Newey2004,VanderLaan2018}. Unfortunately, because ML methods generally seek to optimize out-of-sample performance, they seldom satisfy the latter property.

\subsection{Existing methodological frameworks} \label{section: intro existing methods}

The targeted minimum loss-based estimation (TMLE) framework provides a means of constructing efficient plug-in estimators \citep{VanderLaan2006,VanderLaan2018}. Given an (almost arbitrary) initial ML fit that provides a good estimate of the function-valued features involved, TMLE produces an adjusted fit such that the resulting plug-in estimator has reduced bias and is efficient. This adjustment process is referred to as targeting since a generic estimate of the function-valued features is modified to better suit the goal of estimating the summary of interest. Though TMLE provides a general template for constructing efficient estimators, its implementation requires specialized expertise, namely knowledge of the analytic expression for an influence function of the summary of interest. Influence functions arise in semiparametric efficiency theory and are key to establishing efficiency, but can be difficult to derive. Furthermore, even when an influence function is known analytically, additional expertise is needed to construct a TMLE for a given problem.

Alternative approaches for constructing efficient plug-in estimators have been proposed in the literature, including the use of undersmoothing \citep{Newey1998}, twicing kernels \citep{Newey2004}, and sieves \citep{Chen2007,Newey1997,Shen1997}. These methods neither require knowing an influence function nor performing any targeting of the function-valued feature estimates. Hence, the same fits can be used to simultaneously estimate different summaries of the data-generating distribution, even if these summaries were not pre-specified when obtaining the fit. These approaches also circumvent the difficulties in obtaining an influence function. However, these methods all rely on smoothness conditions on derivatives of the functional features that may be overly stringent. In addition, undersmoothing provides limited guidance on the choice of the tuning parameter; the twicing kernel method requires the use of a higher-order kernel, which may lead to poor performance in small to moderate samples \citep{Marron1994}.

In contrast, under some conditions, sieve estimation can produce a flexible fit with the optimal out-of-sample performance while also yielding an efficient --- and therefore root-$n$-consistent and asymptotically normal --- plug-in estimator \citep{Shen1997}. In this paper, we focus on extensions of this approach. In sieve estimation, we first assume that the unknown function falls in a rich function space, and construct a sequence of approximating subspaces indexed by sample size that increase in complexity as sample size grows. We require that, in the limit, the functions in the subspaces can approximate any function in the rich function space arbitrarily well. These approximating subspaces are referred to as \textit{sieves}. By using an ordinary fitting procedure that optimizes the estimation of the function-valued feature within the sieve, the bias of the plug-in estimator can decrease sufficiently fast as the sieve grows in order for that estimator to be efficient. Thus sieve estimation requires no explicit targeting for the summary of interest.

The series estimator is one of the best known and most widely used sieve techniques. These sieves are taken as the span of the first finitely many terms in a basis that is chosen by the user to approximate the true function well. Common choices of the basis include polynomials, splines, trigonometric series and wavelets, among others. However, series estimators usually require strong smoothness assumptions on derivative of the unknown function in order for the flexible fit to converge at a sufficient rate to ensure the resulting plug-in estimator is efficient. As the dimension of the problem increases, the smoothness requirement may become prohibitive. Moreover, even if the smoothness assumption is satisfied, a prohibitively large sample size may be needed for some series estimators to produce a good fit. For example, if the unknown function is smooth but is a constant over a region, estimation based on a polynomial series can perform poorly in small to moderate samples.

Series estimators may also require the user to choose the number of terms in the series in such a way that results in a sufficient convergence rate. The rates at which the number of terms should grow with sample size have been thoroughly studied (e.g. \citep{Chen2007,Newey1997,Shen1997}). However, these results only provide minimal guidance for applications because there is no indication on how to select the actual number of terms for a given sample size. In practice, the number of terms is the series is often chosen by CV. Upper bounds on the convergence rate of the series estimator as a function of sample size and the number of terms have been derived, and it has been shown that the optimal number of terms that minimizes the bound can also lead to an efficient plug-in estimator \citep{Shen1997}. However, CV tends to select the number of terms that optimizes the actual convergence rate \citep{Vanderlaan2003cv}, which may differ from the number of terms minimizing the derived bound on the convergence rate. Even though the use of CV-tuned sieve estimators has achieved good numerical performance, to the best of our knowledge, there is no theoretical guarantee that they lead to an efficient plug-in estimator.

Two variants of traditional series estimators were proposed in \cite{Bickel2003}. These methods can use two bases to approximate the unknown function-valued features and the corresponding gradient separately, whereas in traditional series estimators, only one basis is used for both approximations. Consequently, these variants may be applied to more general cases than traditional series estimators. However, like traditional series estimators, they also suffer from the inflexibility of the pre-specified bases.

\subsection{Contributions and organization of this article}

In this paper we present two approaches that can partially overcome these shortcomings.
\begin{enumerate}
    \item \textit{Estimating the unknown function with Highly Adaptive Lasso (HAL)} \citep{benkeser2016,VanderLaan2017}.\\
    If we are willing to assume the unknown functions have a finite variation norm, then they may be estimated via HAL. If the tuning parameter is chosen carefully, then we may obtain an efficient plug-in estimator. This method can help overcome the stringent smoothness assumptions on derivatives that are required by existing series estimators, as we discussed earlier.
    
    \item \textit{Using data-adaptive series based on an initial ML fit}.\\
    As long as the initial ML algorithm converges to the unknown function at a sufficient rate, we show that, for certain types of summaries, it is possible to obtain an efficient plug-in estimator with a particular data-adaptive series. The smoothness assumption on the unknown function can be greatly relaxed due to the introduction of the ML algorithm into the procedure. Moreover, for summaries that are highly smooth, we show that the number of terms in the series can be selected by CV.
\end{enumerate}
Although the first approach is not an example of sieve estimation, both approaches are motivated by the sieve literature and can be shown to lead to asymptotically efficient plug-in estimators using the sieve estimation theory derived in \cite{Shen1997}. The flexible fits of the functional features from both approaches can be plugged in for a rich class of estimands.

We remark that, although we do not have to restrict ourselves to the plug-in approach in order to construct an asymptotically efficient estimator, other estimators do not overcome the shortcomings described in Section~\ref{section: intro existing methods} and can have other undesirable properties. For example, the popular one-step correction approach (also called debiasing in the recent literature on high-dimensional statistics) \citep{Pfanzagl1982} constructs efficient estimators by adding a bias reduction term to the plug-in estimator. Thus, it is not a plug-in estimator itself, and as a consequence, one-step estimators may not respect known constraints on the estimand --- for example, bounds on a scalar-valued estimand (e.g., the estimand is a probability and must lie in $[0,1]$) or shape constraints on a vector-valued estimand (e.g., monotonicity constraints). This drawback is also typical for other non-plug-in estimators, such as those derived via estimating equations \citep{VanderLaan2003} and double machine learning \citep{Chernozhukov2017,Chernozhukov2018}. Additionally, as with the other procedures described above, the one-step correction approach requires the analytic expression of an influence function.

Our paper is organized as follows. We introduce the problem setup and notation in Section~\ref{section: setup}. We consider plug-in estimators based on HAL in Section~\ref{section: hal}, data-adaptive series in Section~\ref{section: simple case}, and its generalized version that is applicable to more general summaries in Section~\ref{section: general case}. Section~\ref{section: discussion} concludes with a discussion. Technical proofs of lemmas and theorems (Appendix~\ref{appendix: proof}), simulation details (Appendix~\ref{appendix: simulation}) and other additional details are provided in the Appendix.

\section{Problem setup and traditional sieve estimation review} \label{section: setup}

Suppose we have independent and identically distributed observations $V_1,\ldots,V_n$ drawn from $P_0$. Let $\Theta$ be a class of functions, and denote by $\theta_0 \in \Theta$ a (possibly vector-valued) functional feature of $P_0$ --- for example, $\theta_0$ may be a regression function. Throughout this paper we assume that the generic data unit is $V=(X,Z) \sim P_0$, where $X$ is a (possibly vector-valued) random variable corresponding to the argument of $\theta_0$, and $Z$ may also be a vector-valued random variable. In some cases $V=X$ and $Z$ is trivial. We use $\mathcal{X}$ to denote the support of $X$. The estimand of interest is a finite-dimensional summary $\Psi(\theta_0)$ of $\theta_0$. We consider a plug-in estimator $\Psi(\hat{\theta}_n)$, where $\hat{\theta}_n$ is an estimator of $\theta_0$, and aim for this plug-in estimator to be asymptotically linear, in the sense that $\Psi(\hat{\theta}_n)=\Psi(\theta_0) + n^{-1} \sum_{i=1}^n \text{IF}(V_i) + o_p(n^{-1/2})$ with $\text{IF}$ an influence function satisfying $\expect_{P_0}[\text{IF}(V)]=0$ and $\expect_{P_0}[\text{IF}(V)^2]<\infty$. This estimator is efficient under a nonparametric model if the estimator is also regular. By the central limit theorem and Slutsky's theorem, it follows that $\Psi(\hat{\theta}_n)$ is a CAN estimator of $\Psi(\theta_0)$, and therefore, $\sqrt{n} [\Psi(\hat{\theta}_n) - \Psi(\theta_0)] \overset{d}{\rightarrow} N(0, \expect_{P_0}[\text{IF}(V)^2])$. This provides a basis for constructing valid confidence intervals for $\Psi(\theta_0)$.

We now list some examples of such problems.

\begin{example} \label{example: moments}
    Moments of the conditional mean function \citep{Shen1997}: Let $\theta_0: x \mapsto \expect_{P_0}[Z|X=x]$ be the conditional mean function. The $\kappa$-th moment of $\theta_0(X)$, $X \sim P_0$, namely $\Psi_\kappa(\theta_0)=\expect_{P_0}[\theta_0^\kappa(X)]$, can be a summary of interest. The values of $\Psi_1(\theta_0)$ and $\Psi_2(\theta_0)$ are useful for defining the proportion of $\Var_{P_0}(Z)$ that is explained by $X$, which may be written as $\Var_{P_0}(\theta_0(X))/\Var_{P_0}(Z)$. This proportion is a measure of variable importance \citep{Williamson2017}. Generally, we may consider $\Psi(\theta_0)=\expect_{P_0}[f(\theta_0(X))]$ for a fixed function $f$.
\end{example}

\begin{example} \label{example: average derivative}
    Average derivative \citep{Hardle1989}: Let $X$ follow a continuous distribution on $\real^d$ and $\theta_0: x \mapsto \expect_{P_0}[Z|X=x]$ be the conditional mean function. Let $\theta_0'$ denote the vector of partial derivatives of $\theta_0$. Then $\Psi(\theta_0)=\expect_{P_0}[\theta_0'(X)]$ summarizes the overall (adjusted) effect of each component of $X$ on $Y$. Under certain conditions, we can rewrite $\Psi(\theta_0)=\expect_{P_0}[\theta_0(X) p_0'(X)/p_0(X)]$, where $p_0$ is the Lebesgue density of $X$ and $p_0'$ is the vector of partial derivatives of $p_0$. This expression clearly shows the important role of the Lebesgue density of $X$ in this summary.
\end{example}

\begin{example} \label{example: ATE}
    Mean counterfactual outcome \citep{Rubin1974}:  Suppose that $Z=(A,Y)$ where $A$ is a binary treatment indicator and $Y$ is the outcome of interest. Let $\theta_0: x \mapsto \expect_{P_0}[Y|A=1,X=x]$ be the outcome regression function under treatment value 1. Under causal assumptions, the mean counterfactual outcome corresponding to the intervention that assigns treatment 1 to the entire population can be nonparametrically identified by the G-computation formula $\Psi(\theta_0)=\expect_{P_0}[\theta_0(X)]$.
\end{example}

\begin{example} \label{example: treatment effect heterogeneity}
    Treatment effect heterogeneity measures \citep{Levy2018}: Similarly to Example~\ref{example: ATE}, suppose that $A$ is a binary treatment indicator and $Z$ is the outcome of interest. Let $\theta_0=(\mu_{00},\mu_{01})^\top$, where $\mu_{0a}: x \mapsto \expect_{P_0}[Z|A=a,X=x]$ is the outcome regression function for treatment arm $a\in\{0,1\}$. Then, $\Psi(\theta_0)=\Var_{P_0}(\mu_{01}(X)-\mu_{00}(X))$ is an overall summary of treatment effect heterogeneity.
\end{example}

To obtain an asymptotically linear plug-in estimator, $\hat{\theta}_n$ must converge to $\theta_0$ at a sufficiently fast rate and approximately solve an estimating equation to achieve the small bias property with respect to the summary of interest \citep{Newey2004,VanderLaan2017,VanderLaan2018}. For simplicity, we assume the estimand to be scalar-valued --- when the estimand is vector-valued, we can treat each entry as a separate estimand, and the plug-in estimators of all entries are jointly asymptotically linear if each estimator is asymptotically linear. Therefore, this leads to no loss in generality if the same fits are used for all entries in the summary of interest.

Sieve estimation allows us to obtain an estimator $\Psi(\hat{\theta}_n)$ with the small bias property with respect to $\Psi(\theta_0)$ while maintaining the optimal convergence rate of $\hat{\theta}_n$ \citep{Chen2007,Shen1997}. The construction of sieve estimators is based on a sequence of approximating spaces $\Theta_n$ to $\Theta$. These approximating spaces are referred to as \textit{sieves}. Usually $\Theta_n$ is much simpler than $\Theta$ to avoid over-fitting but complex enough to avoid under-fitting. For example, $\Theta_n$ can be the space of all polynomials with degree $K$ or splines with $K$ knots with $K=K(n) \rightarrow \infty$ as $n \rightarrow \infty$. In this paper, with a loss function $\ell$ such that $\theta_0 \in \argmin_{\theta \in \Theta} \expect_{P_0}[\ell(\theta)(V)]$, we consider estimating $\theta_0$ by minimizing an empirical risk based on $\ell$, i.e., $\hat{\theta}_n \in \argmin_{\theta \in \Theta_n} n^{-1} \sum_{i=1}^{n} \ell(\theta)(V_i)$. Under some conditions, the growth rate of $\Theta_n$ can be carefully chosen so that $\Psi(\hat{\theta}_n)$ is an asymptotically linear estimator of $\Psi(\theta_0)$ while $\hat{\theta}_n$ converges to $\theta_0$ at the optimal rate.

Throughout this paper, for a probability distribution $P$ and an integrable function $f$ with respect to $P$, we define $Pf := \int f(v) dP(v)=\expect_P[f(V)]$. We use $P_n$ to denote the empirical distribution. We take $\langle \cdot, \cdot \rangle$ to be the $L^2(P_0)$-inner product, i.e., $\langle \theta_1, \theta_2 \rangle=P_0 (\theta_1 \theta_2)$, where $L^2(P_0)$ is the set of real-valued $P_0$-squared-integrable functions defined on the support of $P_0$. When the functions are vector-valued, we take $\langle \theta_1, \theta_2 \rangle=P_0 (\theta_1^\top \theta_2)$. We use $\| \cdot \|$ to denote the induced norm of $\langle \cdot, \cdot \rangle$. We assume that $\Theta \subseteq L^2(P_0)$. We remark that we have committed to a specific choice of inner product and norm to fix ideas; other inner products can also be adopted, and our results will remain valid upon adaptation of our upcoming conditions. We discuss this explicitly via a case study in Appendix~\ref{appendix: change norm}.

For the methods we propose in this article, we assume that $\Theta$ is convex. Throughout this paper, we will further require a set of conditions similar to those in \cite{Shen1997}. For any $\theta \in \Theta$, let $\ell_{0}'[\theta-\theta_0](v) := \lim_{\delta \rightarrow 0} [\ell(\theta_0 +\delta(\theta-\theta_0))(v) - \ell(\theta_0)(v)]/\delta$ be the G\^{a}teaux derivative of $\ell$ at $\theta_0$ in the direction $\theta-\theta_0$ and $r[\theta-\theta_0](v) := \ell(\theta)(v) - \ell(\theta_0)(v) - \ell_{0}'[\theta-\theta_0](v)$ be the corresponding remainder.
\begin{conditionA}[Linearity and boundedness of G\^{a}teaux derivative operator of loss function]
    \label{Adloss}
    For all $\theta \in\Theta$, $\ell_{0}'[\theta-\theta_0]$ exists and $\ell_{0}'[\theta-\theta_0](v) - P_0 \ell_{0}'[\theta-\theta_0]$ is linear and bounded in $\theta-\theta_0$.
\end{conditionA}
\begin{conditionA}[Local quadratic behavior of loss function]
    \label{Aquadraticloss}
    There exists a constant $\alpha_{0,\ell} \in (0,\infty)$ such that, for all $\theta \in \Theta$ such that $P_0 \{ \ell(\theta) - \ell(\theta_0) \}$ or $\| \theta - \theta_0 \|$ is sufficiently small, it holds that $P_0 \{ \ell(\theta) - \ell(\theta_0) \} = \alpha_{0,\ell} \| \theta - \theta_0 \|^2 /2 + o(\| \theta - \theta_0 \|^2)$.
\end{conditionA}

\begin{remark}
    We now present an equivalent form of \ref{Aquadraticloss} that may be easier to verify in practice. For all $\theta\in \Theta\backslash\{\theta_0\}$, define $h_\theta:=(\theta-\theta_0)/\|\theta-\theta_0\|$ and $a_{\theta}:= \frac{d^2}{d\delta^2} P_0 \ell(\theta_0 + \delta h_\theta) |_{\delta=0}$. Requiring Condition~\ref{Aquadraticloss} is equivalent to requiring that $a_{\theta_1}=a_{\theta_2}$ for all $\theta_1,\theta_2\in\Theta\backslash\{\theta_0\}$ and that
    \begin{align*}
    \sup_{\theta\in\Theta} \left|P_0 \ell(\theta_0 + \delta h_\theta) - P_0 \ell(\theta_0) - \frac{a_\theta}{2}\right|&= o(\delta^2).
    \end{align*}
    Moreover, if \ref{Aquadraticloss} holds, then, for any $\theta\in\Theta\backslash\{\theta_0\}$, it is true that $\alpha_{0,\ell}=a_\theta$.
\end{remark}

A large class of loss functions satisfy Conditions~\ref{Adloss} and \ref{Aquadraticloss}. For example, in the regression setting where $Z$ is the outcome, the squared-error loss $\ell(\theta): v \mapsto [z-\theta(x)]^2$ and the logistic loss $\ell(\theta): v \mapsto -z \theta(x) + \log\{1+\exp(\theta(x))\}$ both satisfy these conditions; a negative working log-likelihood usually also satisfies these conditions. In Examples~\ref{example: moments}--\ref{example: treatment effect heterogeneity}, the unknown functions are all conditional mean functions, which can be estimated with the above loss functions. Thus, Conditions~\ref{Adloss} and \ref{Aquadraticloss} hold. Examples~\ref{example: ATE} and \ref{example: treatment effect heterogeneity} require a slight modification discussed in more details in Appendix~\ref{appendix: change norm}. We also note that Condition~\ref{Aquadraticloss} is sufficient for Condition~B in \cite{Shen1997}.

\begin{conditionA}[Differentiability of summary of interest]
    \label{AdPsi}
    $\Psi_{\theta_0}'[\theta-\theta_0] := \lim_{\delta \rightarrow 0} [\Psi(\theta_0+\delta (\theta-\theta_0)) - \Psi(\theta_0)]/\delta$ exists for all $\theta \in \Theta$ and is a linear bounded operator.
\end{conditionA}
If Condition \ref{AdPsi} holds, then, by the Riesz representation theorem, $\Psi_{\theta_0}'[\theta-\theta_0] = \langle \theta-\theta_0, \dot{\Psi} \rangle$ for a gradient function $\dot{\Psi}=\dot{\Psi}_{\theta_0}$ in the completion of the space spanned by $\Theta-\theta_0 := \{x \mapsto \theta(x)-\theta_0(x): \theta \in \Theta\}$.
\begin{conditionA}[Locally quadratic remainder]
    \label{APsiremainder}
    There exists a constant $C>0$ so that, for all $\theta$ with sufficiently small $\| \theta - \theta_0 \|$, it holds that
    $$|\Psi(\theta)-\Psi(\theta_0)-\Psi_{\theta_0}'[\theta-\theta_0]| \leq C \| \theta - \theta_0 \|^2.$$
\end{conditionA}
The above condition states that the remainder of the linear approximation to $\Psi$ is locally bounded by a quadratic function.

Conditions~\ref{AdPsi} and \ref{APsiremainder} hold for Examples~\ref{example: moments}--\ref{example: treatment effect heterogeneity}. For the generalized moment of the conditional mean function in Example~\ref{example: moments}, it holds that $\dot{\Psi}=f' \circ \theta_0$. For the average derivative of the conditional mean function in Example~\ref{example: average derivative}, it holds that $\dot{\Psi}=p_0'/p_0$. For the average treatment effect and the treatment effect heterogeneity measure in Examples~\ref{example: ATE} and \ref{example: treatment effect heterogeneity}, as we show in Appendix~\ref{appendix: change norm}, $\dot{\Psi}$ also exists and depends on the propensity score function $x \mapsto P_0(A=1 \mid X=x)$.

\section{Estimation with Highly Adaptive Lasso} \label{section: hal}

\subsection{Brief review of Highly Adaptive Lasso} \label{section: hal review}

Recently, the Highly Adaptive Lasso (HAL) was proposed as a flexible ML algorithm that only requires a mild smoothness condition on the unknown function and has a well-described implementation \citep{benkeser2016,VanderLaan2017}. In this subsection, we briefly review HAL. We first heuristically introduce its definition and desirable properties, and then introduce the definition and implementation more formally. For ease of presentation, for the moment, we assume that $\theta_0$ is real-valued.

In HAL, $\theta_0$ is assumed to fall in the class of c\`{a}dl\`{a}g functions (right-continuous with left limits) defined on $\mathcal{X} \subseteq \real^d$ with variation norm bounded by a finite constant $M$. In this section, we denote this function class by $\Theta_{\vnorm,M}$. The variation norm of a c\`{a}dl\`{a}g function $\theta$, denoted by $\| \theta \|_{\vnorm}$, characterizes the total variability of $\theta$ as its argument ranges over the domain, so $\| \cdot \|_{\vnorm}$ is a global smoothness measure and $\Theta_{\vnorm,M}$ is a large function class that even contains functions with discontinuities. Fig.~\ref{Fcadlag} presents some examples of univariate c\`{a}dl\`{a}g functions with finite variation norms for illustration. Because $\Theta_{\vnorm,M}$ is a rich class, it can be plausible that $\theta_0 \in \Theta_{\vnorm,M}$ for some $M < \infty$. The HAL estimator of $\theta_0$ is then $\hat{\theta}_n=\hat{\theta}_{n,M} \in \argmin_{\theta \in \Theta_{\vnorm,M}} n^{-1} \sum_{i=1}^n \ell(\theta)(V_i)$. Under this assumption, it has been shown that $\| \hat{\theta}_n - \theta_0 \| = o_p(n^{-1/4})$ regardless of the dimension of $X$ under additional mild conditions \citep{VanderLaan2017}. Thus, estimation with HAL replaces the usual smoothness requirement on derivatives of traditional series estimators by a requirement on global smoothness, namely $\theta_0 \in \Theta_{\vnorm,M}$ for some $M$.

\begin{figure}[h]
    \centering
    \includegraphics[scale=0.9]{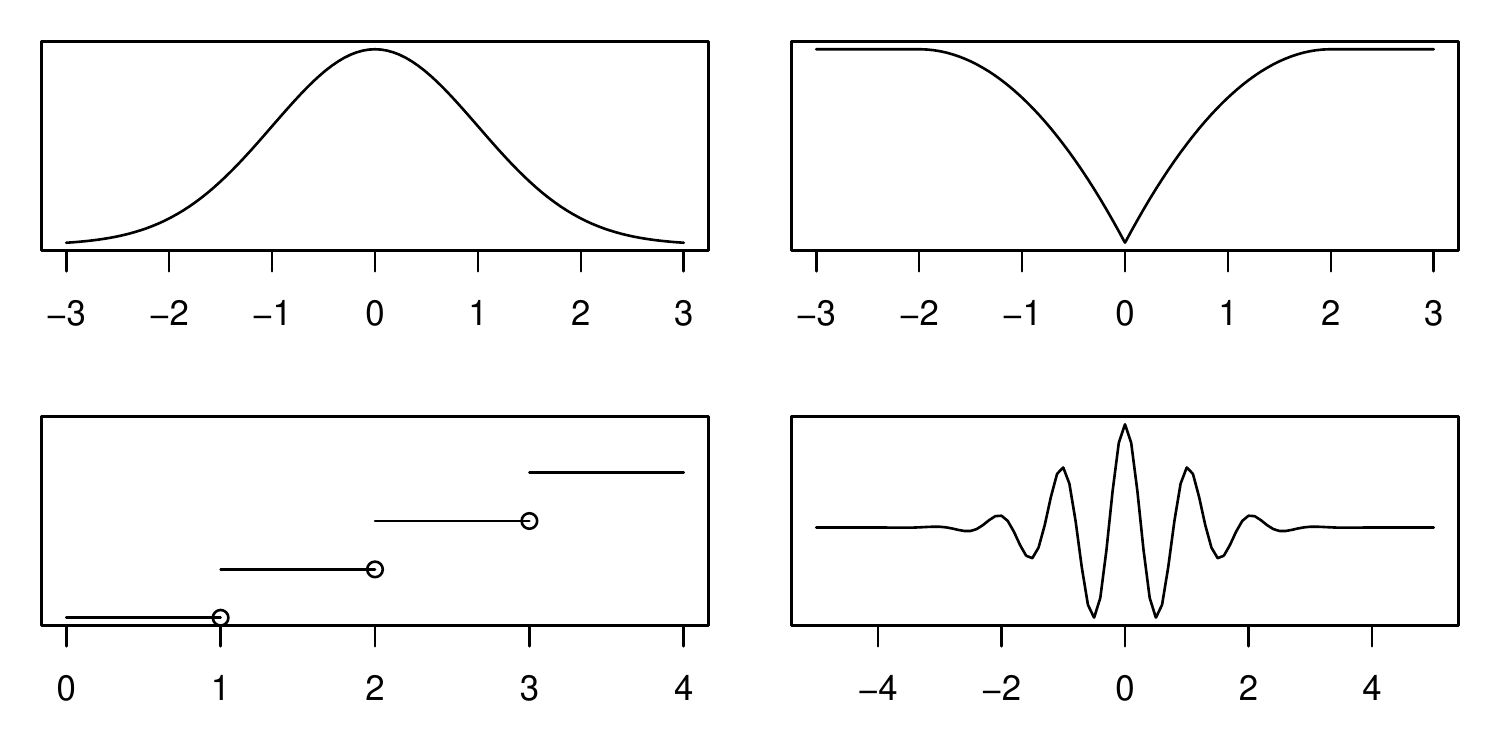}
    \caption{Examples of univariate c\`{a}dl\`{a}g functions with finite variation norms. The top-left, top-right, bottom-left and bottom-right plots present the standard normal density function, a minimax concave penalty function \citep{Zhang2010}, a step function and the real part of a Morlet wavelet \citep{Mallat2009} respectively.}
    \label{Fcadlag}
\end{figure}

We next formally present the definition of variation norm of a c\`{a}dl\`{a}g function $\theta: [ x^{(\ell)},x^{(u)} ] \subseteq \real^d \rightarrow \real$. Here, $x^{(\ell)}$ and $x^{(u)}$ are vectors in $\real^d$; with $\leq$ being entrywise, $[ x^{(\ell)},x^{(u)} ] := \{x \in \real^d: x^{(\ell)} \leq x  \leq x^{(u)} \}$.

For any nonempty index set $s \subseteq \{1,2,\ldots,d\}$ and any $x=(x_1,x_2,\ldots,x_d) \in [ x^{(\ell)},x^{(u)} ]$, we define $x_s := \{x_j: j \in s\}$ and $x_{-s} := \{x_j: j \in \{1,2,\ldots,d\} \setminus s \}$ to be entries of $x$ with indices in and not in $s$ respectively. We defined the $s$-section of $\theta$ as $\theta_s := \theta(x_1 \ind(1 \in s), x_2 \ind(2 \in s), \ldots, x_d \ind(d \in s))$. We can subsequently obtain the following representation of $\theta$ at any $x \in [ x^{(\ell)},x^{(u)} ]$ in terms of sums and integrals of the variation of $s$-sections of $\theta$ \citep{gill1993}:
$$\theta(x) = \theta ( x^{(\ell)} ) + \sum_{s \in \{1,\ldots,d\}, s \neq \emptyset} \int_{( x^{(\ell)},x ]} \theta_s(d \tilde{x}).$$
The variation norm is then subsequently defined as
$$\| \theta \|_\vnorm := | \theta ( x^{(\ell)} ) | + \sum_{s \in \{1,\ldots,d\}, s \neq \emptyset} \int_{( x^{(\ell)},x^{(u)} ]} | \theta_s(d \tilde{x}) |.$$
We refer to \cite{benkeser2016} and \cite{VanderLaan2017} for more details on variation norm. Notably, this notion of variation norm coincides with that of Hardy and Krause \citep{owen2005}.

We finally briefly introduce the algorithm to compute a HAL estimator. It can be shown that an empirical risk minimizer in $\Theta_{\vnorm,M}$ is a step function that only jumps at sample points, namely
$$x \mapsto \beta_0 + \sum_{s \subseteq \{1,\ldots,d\}, s \neq \emptyset} \sum_{j=1}^n \ind(X_{j,s} \leq x_s) \beta_{s,j}.$$
Here, $\beta_0$ and all $\beta_{s,j}$ are real numbers. To find an empirical risk minimizer in $\Theta_{\vnorm,M}$ in the above form, we may solve the following optimization problem:
\begin{align*}
& \min_{\theta} \sum_{i=1}^n \ell(\theta)(V_i) \\
\text{subject to} \quad & \theta: x \mapsto \beta_0 + \sum_{s \subseteq \{1,\ldots,d\}, s \neq \emptyset} \sum_{j=1}^n \ind(X_{j,s} \leq x_s) \beta_{s,j} \\
& |\beta_0| + \sum_{s \subseteq \{1,\ldots,d\}, s \neq \emptyset} \sum_{j=1}^n |\beta_{s,j}| \leq M.
\end{align*}
The constraint imposes an upper bound on the $\ell_1$ norm of a vector. Therefore, for common loss functions, we may use software for LASSO regression \citep{Tibshirani1996}. For example, if the loss function is the squared-error loss, then we may run a LASSO linear regression to obtain a HAL estimate.

\subsection{Estimation with an oracle tuning parameter} \label{section: hal oracle}

In this section, we consider plug-in estimators based on HAL. For ease of illustration, for the rest of this section, we consider scalar-valued $\Psi$, and will discuss vector-valued $\Psi$ only at the end of this subsection. We further introduce the following conditions needed to establish that the HAL-based plug-in estimator is efficient.
\begin{conditionB}[C\`{a}dl\`{a}g functions]
    \label{Acadlag}
    $\theta_0$ and $\dot{\Psi}$ are c\`{a}dl\`{a}g.
\end{conditionB}
\begin{conditionB}[Bound on variation norm]
    \label{AM}
    For some $M<\infty$, $\| \theta_0 \|_{\vnorm} + \| \dot{\Psi} \|_{\vnorm} \leq M$.
\end{conditionB}

Condition~\ref{AM} ensures that certain perturbations of $\theta_0$ still lie in $\Theta_{\vnorm,M}$, a crucial requirement for proving the asymptotic linearity of our proposed plug-in estimator. In addition, since $\dot{\Psi}$ may depend on components of $P_0$ other than $\theta_0$ as in Examples~\ref{example: average derivative}--\ref{example: treatment effect heterogeneity}, Conditions~\ref{Acadlag}--\ref{AM} may also impose conditions on these components.

In this section, we fix an $M$ that satisfies Condition~\ref{AM}. Additional technical conditions can be found in Appendix~\ref{section: HAL additional regularity conditions}. Let $\hat{\theta}_n=\hat{\theta}_{n,M} \in \argmin_{\theta \in \Theta_{\vnorm,M}} n^{-1} \sum_{i=1}^n \ell(\theta)(V_i)$ denote the HAL fit obtained using the bound $M$ in Condition~\ref{AM}.

We note that $\hat{\theta}_n$ is not a typical sieve estimator because $M$ is fixed and there is no explicit sequence of growing approximating spaces $\Theta_n$. Nevertheless, we may view this method as a special case of sieve estimation with degenerate sieves $\Theta_n=\Theta_{\vnorm,M}$ for all $n$. This allows us to utilize existing results \citep{Shen1997} to show the asymptotic linearity and efficiency of the plug-in estimator based on $\hat{\theta}_n$. We next formally present this result.

\begin{theorem}[Efficincy of plug-in estimator]
    \label{THALefficiency}
    Under Conditions~\ref{Adloss}--\ref{APsiremainder} and \ref{Acadlag}--\ref{AHALfinitevar}, $\Psi(\hat{\theta}_n)$ is an asymptotically linear estimator of $\Psi(\theta_0)$ with the influence function being $v \mapsto \alpha_{0,\ell}^{-1} \{-\ell_{0}'[\dot{\Psi}] (v) + \expect_{P_0} [ \ell_{0}'[\dot{\Psi}] (V) ]\}$, that is,
    $$\Psi(\hat{\theta}_n) = \Psi(\theta_0) + \frac{1}{n} \sum_{i=1}^n \alpha_{0,\ell}^{-1} \left\{ -\ell_{0}'[\dot{\Psi}] (V_i) + \expect_{P_0} \left[ \ell_{0}'[\dot{\Psi}] (V) \right] \right\}+o_p(n^{-1/2}).$$
    As a consequence, $\sqrt{n} [\Psi(\hat{\theta}_n)-\Psi(\theta_0)] \overset{d}{\rightarrow} \text{N}(0, \xi^2 )$ with $\xi^2 := \Var_{P_0}(\ell_{0}'[\dot{\Psi}] (V))/\alpha_{0,\ell}^2$. In addition, under Conditions~\ref{Acloseminimizer} and \ref{Aempiricalprocesslike} in Appendix~\ref{section: regularity additional conditions}, $\Psi(\hat{\theta}_n)$ is efficient under a nonparametric model.
\end{theorem}

We note that, for HAL to achieve the optimal convergence rate, we only need that $M \geq \| \theta_0 \|_{\vnorm}$ \citep{benkeser2016,VanderLaan2017}. The requirement of a larger $M$ imposed by Condition~\ref{AM} resembles undersmoothing \citep{Newey1998}, as using a larger $M$ would result in a fit that is less smooth than that based on the CV-selected bound. The $L^2(P_0)$-convergence rate of the flexible fit using the larger bound remains the same, but the leading constant may be larger. This is in contrast to traditional undersmoothing, which leads to a fit with a suboptimal rate of convergence.

Under some conditions, the following lemma provides a loose bound on $\| \dot{\Psi} \|_{\vnorm}$ in the case that $\dot{\Psi}$ has a particular structure. Such a bound can be used to select an appropriate bound on variation norm that satisfies Condition~\ref{AM}.
\begin{lemma}
    \label{Lvarnorm}
    Suppose that $\dot{\Psi}=\dot{\psi} \circ \theta_0$, where $\dot{\psi}: \real \rightarrow \real$ is differentiable. Let $x^{(\ell)}=\sup \{x: P_0(X \geq x) = 1\}$ where $\sup$ and $\geq$ are entrywise. Assume that $\theta_0$ is differentiable. If each of $\| \theta_0 \|_{\vnorm}$, $|\dot{\Psi}(x^{(\ell)})|$ and $B := \sup_{z: |z| \leq \| \theta_0 \|_{\vnorm}} |\dot{\psi}'(z)|$ is finite, then $\| \dot{\Psi} \|_{\vnorm} \leq B \| \theta_0 \|_{\vnorm} + |\dot{\Psi}(x^{(\ell)})|$. Hence, $\| \theta_0 \|_{\vnorm} + \| \dot{\Psi} \|_{\vnorm} \leq (B+1) \| \theta_0 \|_{\vnorm} + |\dot{\Psi}(x^{(\ell)})| < \infty$.
\end{lemma}

As we discussed at the end of Section~\ref{section: setup}, such structures as $\dot{\Psi}=\dot{\psi} \circ \theta_0$ are common, especially if we augment $\theta_0$ to include other implicitly relevant components of $P_0$. For example, in Example~\ref{example: average derivative}, we may augment $\theta_0$ with $p_0$ and $p_0'$; in Examples~\ref{example: ATE} and \ref{example: treatment effect heterogeneity}, we may augment $\theta_0$ with the propensity score function.

When $\theta_0$ is $\real^q$-valued, $\theta_0$ can often be viewed as a collection of $q$ real-valued variation-independent functions $\eta_{10}, \ldots, \eta_{q0}$. In this case, we can define $\Theta_{\vnorm,M} = \{ (\eta_1, \ldots, \eta_q): \eta_j \text{ is c\`{a}dl\`{a}g}, \|\eta_j\|_{\vnorm} \leq M_j, j=1,\ldots,q \}$ for a positive vector $M=(M_1,\ldots,M_q)$. The subsequent arguments follow analogously, where now each $\eta_j$ is treated as a separate function.

We remark that an undersmoothing condition such as \ref{AM} appears to be necessary for a HAL-based plug-in estimator to be efficient. We illustrate this numerically in Section~\ref{section: hal_sim}. The choice of a sufficiently large bound $M$ required by Theorem~\ref{THALefficiency} is by no means trivial, since this choice requires knowledge that the user may not have. Nevertheless, this result forms the basis of the data-driven method that we propose in Section~\ref{section: hal_sim} for choosing $M$. We also remark that, if we wish to plug in the same $\hat{\theta}_n$ based on HAL for a rich estimands, the chosen bound $M$ needs to be sufficiently large for all estimands of interest.

Another method to construct efficient plug-in estimators based on HAL has been independently developed \citep{VanderLaan2019}. Unlike our approach based on sieve theory, in this work, the authors directly analyzed the first-order bias of the plug-in estimator using influence functions. In terms of ease of implementation, their method requires specifying a constant involved in a threshold of the empirical mean of the basis functions, which may be difficult to specify in applications. Our approach in Section~\ref{section: hal_sim} may also require specifying an unknown constant to obtain a valid upper bound on $\|\dot{\Psi}\|_{\vnorm}$, but in some cases the constant may be set to zero, and our simulation suggests that the performance is not sensitive to the choice of the constant.

\subsection{Data-adaptive selection of the tuning parameter} \label{section: hal_sim}

Since it is hard to prespecify a bound $M$ on the variation norm that is sufficiently large to satisfy Condition~\ref{AM} but also sufficiently small to avoid overfitting for a given data set, it is desirable to select $M$ in a data-adaptive manner. A seemingly natural approach makes use of $k$-fold CV. In particular, for each candidate bound $M$, partition the data into $k$ folds of approximately equal size ($k$ is fixed and  does not depend on $n$), in each fold evaluate the performance of the HAL estimator fitted on all other folds based on this candidate $M$, and use the candidate bound $M_n$ with the best average performance across all folds to obtain the final fit. It has been shown that $\hat{\theta}_{n,M_n}$ can achieve the optimal convergence rate under mild conditions \citep{Vanderlaan2003cv}, but $M_n$ appears not to satisfy Condition \ref{AM} in general. In particular, the derived bound on $\| \hat{\theta}_n - \theta_0 \|$ relies on an empirical process term, namely $\sup_{\theta \in \Theta_{\vnorm,M}} |(P_n-P_0) \{ \ell(\theta) - \ell(\theta_0) \}|$, and a larger $M$ implies a larger space $\Theta_{\vnorm,M}$. Therefore, the bound on $\| \hat{\theta}_n - \theta_0 \|$ grows with $M$. Because $k$-fold CV seeks to optimize out-of-sample performance, $M_n$ generally appears to be close to $\| \theta_0 \|_{\vnorm}$ and not sufficiently large to obtain an efficient plug-in estimator.

To avoid this issue with the CV-selected bound, we propose a method that takes inspiration from $k$-fold CV, but modifies the bound so that it is guaranteed to yield an efficient plug-in estimator for $\Psi(\theta_0)$. This method may require the analytic expression for $\dot{\Psi}$. In Sections~\ref{section: simple case}~and~\ref{section: general case}, we present methods that do not require this knowledge.
\begin{enumerate}
    \item Derive an upper bound on $\| \dot{\Psi} \|_{\vnorm}$. This bound is a non-decreasing function of the variation norms of functions that can be learned from data (e.g., using Lemma~\ref{Lvarnorm}). In other words, find a non-decreasing function $F$ such that $\| \dot{\Psi} \|_{\vnorm} \leq F(\| \eta_{10} \|_{\vnorm}, \ldots, \| \eta_{q0} \|_{\vnorm})$ for unknown functions $\eta_{10},\ldots,\eta_{q0}$ that can be assumed to be c\`{a}dl\`{a}g with finite variation norm and can be estimated with HAL.
    \item Estimate $\theta_0,\eta_{10},\ldots,\eta_{q0}$ by HAL with $k$-fold CV, and denote the CV-selected bounds for these functions by $M_n,M_{1n},\ldots,M_{qn}$.
        \item For a small $\epsilon>0$, use the bound $M_n + \epsilon + F(M_{1n}+\epsilon,\ldots,M_{qn}+\epsilon)$ to estimate $\theta_0$ with HAL and plug in the fit. We refer to this step of slightly increasing the bounds as \textit{$\epsilon$-relaxation}.
\end{enumerate}
It follows from Lemma~\ref{Lcvbound} in the Appendix that this method would yield a sufficiently large bound with probability tending to one. In practice, it is desirable for the bound derived on $\| \dot{\Psi} \|_{\vnorm}$ to be relatively tight to avoid choosing an overly large bound that leads to overfitting in small to moderate samples. We remark that multiplying by $1+\epsilon$ rater than adding $\epsilon$ to each argument also leads to a valid choice for the bound; that is, the bound $M_n (1+\epsilon) + F(M_{1n} (1+\epsilon), \ldots, M_{qn} (1+\epsilon))$ is also sufficiently large with probability tending to one. In practice, the user may increase each CV-selected bound by, for example, 5\% or 10\%. Although it is more natural and convenient to directly use $M_n + F(M_{1n},\ldots,M_{qn})$ as the bound, we have only been able to prove the result with a small $\epsilon$-relaxation. However, if the bound is loose and $F$ is continuous, we can show that $\epsilon$-relaxation is unnecessary. The formal argument can be found after Lemma~\ref{Lcvbound} in the Appendix.

As for methods based on knowledge of an influence function, deriving $\dot{\Psi}$ and a bound for its variation norm requires some expertise, but in some cases this task can be straightforward. The derivation of an influence function is typically based on a fluctuation in the space of distributions, but in many cases, the relation between such fluctuations and the summary of interest is implicit and difficult to handle. In contrast, the derivation of $\dot{\Psi}$ is based on a fluctuation of $\theta_0$, and the summary of interest explicitly depends on $\theta_0$. As a consequence, it can be simpler to derive $\dot{\Psi}$ than to derive an influence function. For example, for the summary $\Psi_\kappa(\theta_0)=P_0 \theta_0^\kappa$ in Example~\ref{example: moments}, we find that $\dot{\Psi}_\kappa=\kappa \theta_0^{\kappa-1}$ by straightforward calculation, whereas the influence function given in Theorem~\ref{THALefficiency} is more difficult to directly derive analytically.

We illustrate the fact that $M_n$ may not be sufficiently large and show that our proposed method resolves this issue via a simulation study in which $\theta_0: x \mapsto \expect_{P_0}[Y|X=x]$ and $\Psi: \theta_0 \mapsto P_0 \theta_0^2$. We compare the performance of the plug-in estimators based on the 10-fold CV-selected bound on variation norm (M.cv), the bound derived from the analytic expression of $\dot{\Psi}$ with and without $\epsilon$-relaxation (M.gcv+ and M.gcv respectively), and a sufficiently large oracle choice satisfying Condition~\ref{AM} (M.oracle). We  According to Lemma~\ref{Lvarnorm}, M.oracle is $3 \| \theta_0 \|_{\vnorm}$ and M.gcv is 3$\times$M.cv. We also investigate the performance of 95\% Wald CIs based on the influence function. For each resulting plug-in estimator, we investigate the following quantities: $n \cdot \text{MSE}$, $\sqrt{n} \cdot |\text{bias}|$ and CI coverage. More details of this simulation are provided in Appendix~\ref{appendix: simulation}. In theory, for an efficient estimator, we should find that $n \cdot \text{MSE}$ tends to a constant (the variance of the influence function $\xi^2 := P_0 \text{IF}^2$), $\sqrt{n} \cdot |\text{bias}|$ tends to $0$, and 95\% Wald CIs have approximately 95\% coverage.

We report performance summaries in Fig~\ref{Fmse_bias_hal} and Table~\ref{TableCI_hal} with this criterion, from which it appears that the plug-in estimators with M.oracle and M.gcv+ achieve efficiency, while the plug-in estimator based on M.cv does not. The desirable performance of M.oracle and M.gcv+ agrees with the available theory, whereas the poor performance of M.cv suggests that cross-validation may not yield a valid choice of variation norm in general. Interestingly, M.gcv performs similarly to M.oracle and M.gcv+. We conjecture that using an $\epsilon$-relaxation is unnecessary in this setting. In Fig~\ref{FM_hal}, we can also see that M.cv tends to $\| \theta_0 \|_{\vnorm}$ and has a high probability of being less than M.oracle. Therefore, this simulation suggests that using a sufficiently large bound --- in particular, a bound larger than the CV-selected bound --- may be necessary and sufficient for the plug-in estimator to achieve efficiency.

\begin{figure}[ht!]
    \centering
    \includegraphics[scale=0.65]{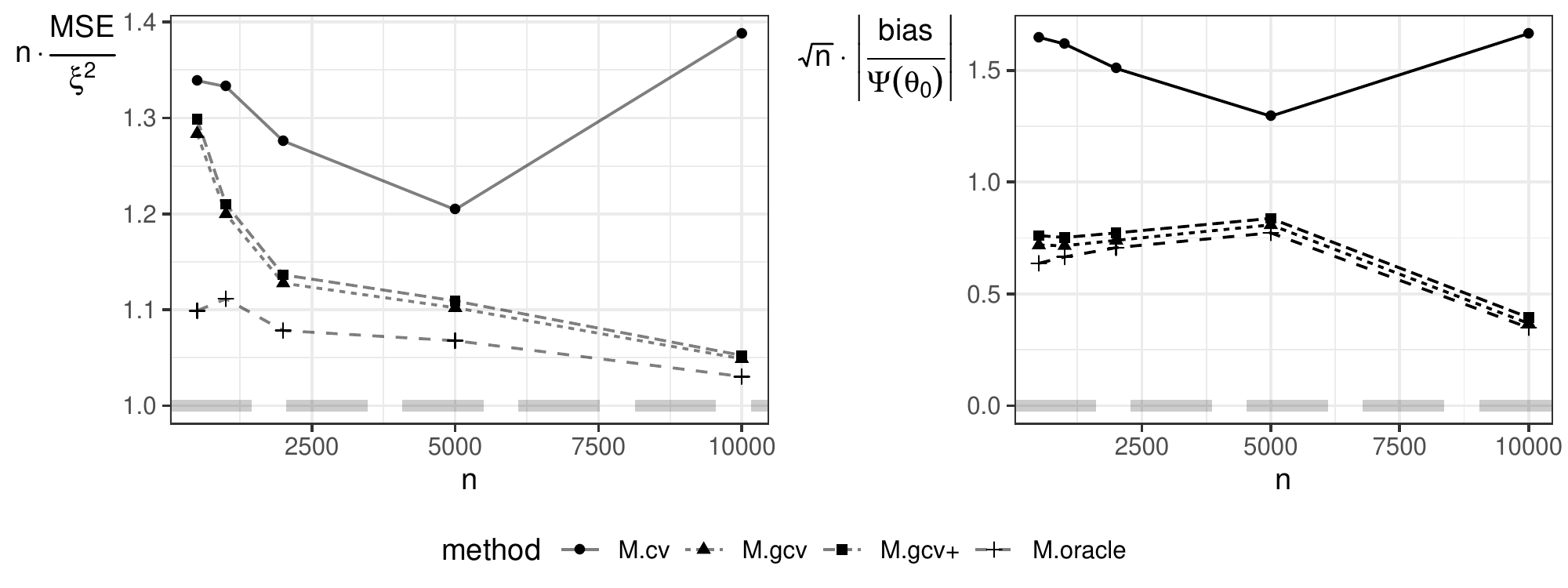}
    \caption{The relative MSE, $n \cdot \text{MSE}/\xi^2$, and the relative absolute bias, $\sqrt{n} \cdot |\text{bias}/\Psi(\theta_0)|$, of the plug-in estimator of $\Psi(\theta_0)=P_0 \theta_0^2$ based on HAL for an oracle choice of the bound on variation norm (M.oracle), the 10-fold CV-selected bound (M.cv), a bound based on M.cv and analytic expression of $\dot{\Psi}$ without and with $\epsilon$-relaxation (M.gcv and M.gcv+ respectively). $\xi^2 := P_0 \text{IF}^2$ is the asymptotic variance that the $n \cdot \text{MSE}$ of an AL estimator should converge to. Note that the $n \cdot \text{MSE}$ for M.oracle, M.gcv and M.gcv+ tends to $\xi^2$ but that for M.cv does not.}
    \label{Fmse_bias_hal}
\end{figure}

\begin{table}[ht!]
    \centering
    \caption{Coverage probability of 95\% Wald CI of the plug-in estimator of $\Psi(\theta_0)=P_0 \theta_0^2$ based on HAL for an oracle choice of the bound on variation norm (M.oracle), the 10-fold CV-selected bound (M.cv), a bound based on M.cv and analytic expression of $\dot{\Psi}$ without and with $\epsilon$-relaxation (M.gcv and M.gcv+ respectively). The CI is constructed based on the influence function. The coverage for M.oracle, M.gcv and M.gcv+ is approximately 95\%, but that for M.cv is not.}
    \label{TableCI_hal}
    \begin{tabular}{rrrrr}
        \hline
        n & M.cv & M.gcv & M.gcv+ & M.oracle \\ 
        \hline
        500 & 0.87 & 0.96 & 0.96 & 0.97 \\ 
        1000 & 0.87 & 0.97 & 0.97 & 0.97 \\ 
        2000 & 0.90 & 0.95 & 0.95 & 0.96 \\ 
        5000 & 0.93 & 0.95 & 0.95 & 0.95 \\ 
        10000 & 0.89 & 0.95 & 0.95 & 0.95 \\ 
        \hline
    \end{tabular}
\end{table}

\begin{figure}[ht!]
    \centering
    \includegraphics[scale=0.65]{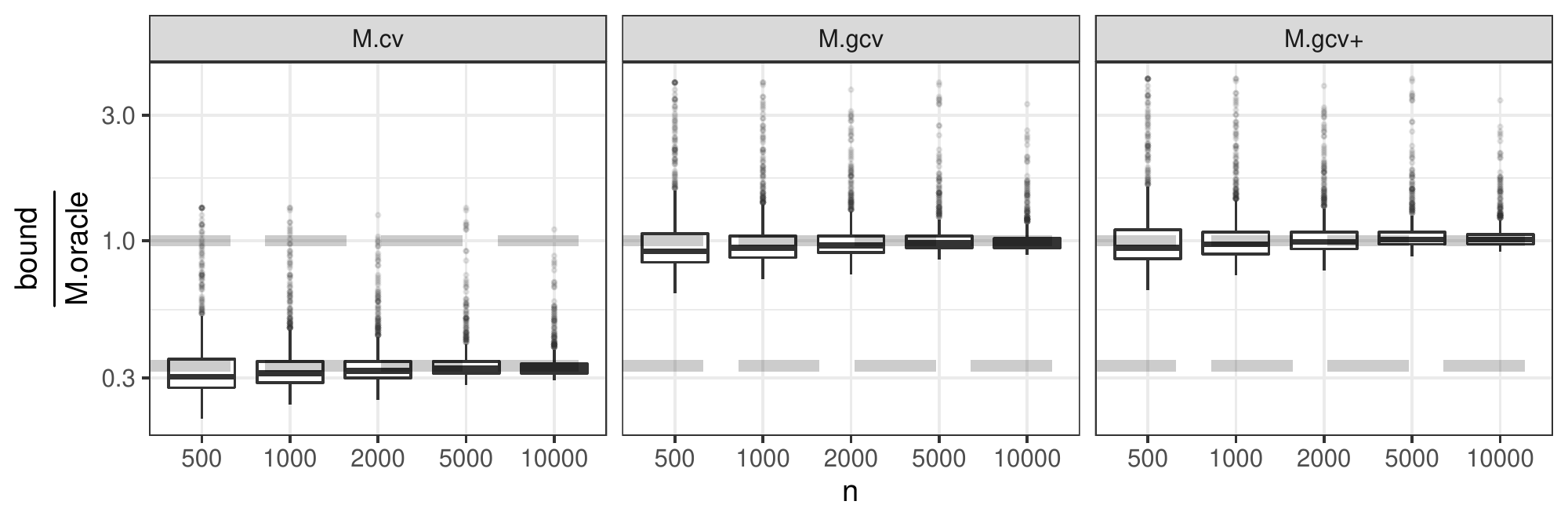}
    \caption{A boxplot of the ratio of bounds based on 10-fold CV and M.oracle. The horizontal gray thick dashed lines are $1$ and $1/3$. The y-axis is scaled based on logarithm for readability. There is a high probability that M.cv is much smaller than M.oracle; M.cv tends to the variation norm of the function being estimated, $\| \theta_0 \|_{\vnorm}$, corresponding to $1/3$ of M.oracle. Enlarging M.cv according to the analytic expression of $\dot{\Psi}$ with $\epsilon$-relaxation results in sufficiently large bounds. The enlargement without $\epsilon$-relaxation appears to have similar performance.}
    \label{FM_hal}
\end{figure}

\section{Data-adaptive series} \label{section: simple case}

\subsection{Proposed method} \label{section: simple case method}

For ease of illustration, we consider the case that $\Psi$ is scalar-valued in this section. As we will describe next, our proposed estimation procedure for function-valued features does not rely on $\Psi$ and hence can be used for a class of summaries.

Suppose that $\Theta$ is a vector space of $\real^q$-valued functions equipped with the $L^2(P_0)$-inner product. Further, suppose that $\dot{\Psi}=\dot{\psi} \circ \theta_0$ for some function $\dot{\psi}: \real^q \rightarrow \real^q$. This holds, for example, when $\Psi: \theta \mapsto P_0 (f \circ \theta)$ for a fixed differentiable function $f$ in Example~\ref{example: moments}. In this case, $\dot{\Psi}=f' \circ \theta_0$ and hence $\dot{\psi}=f'$. Particularly useful examples include Examples~\ref{example: moments} and \ref{example: treatment effect heterogeneity}. For now we assume that the marginal distribution of $X$ is known so that we only need to estimate $\theta_0$ for this summary. We will address the more difficult case in which the marginal distribution of $X$ is unknown in Section~\ref{section: generalized simple case}.

Let $\theta_n^0$ be a given initial flexible ML fit of $\theta_0$ and consider the data-adaptive sieve-like subspaces based on $\theta_n^0$, $\Theta_n := \Theta_{n,\theta_n^0} := \text{Span}\{ \phi_1, \phi_2, \ldots, \phi_K \} \circ \theta_n^0$, where $\phi_1,\phi_2,\ldots$ are $\real^q$-valued basis functions in a series defined on $\real^q$ and $K=K(n)$ is a deterministic number of terms in the series --- we will consider selecting $K$ via CV in Section \ref{section: CV}. Let $\theta_n^*=\theta_n^*(\theta_n^0) \in \argmin_{\theta \in \Theta_n} n^{-1} \sum_{i=1}^{n} \ell(\theta)(V_i)$ denote the series estimator within this data-adaptive sieve-like subspace that minimizes the empirical risk. We propose to use $\Psi(\theta_n^*)$ to estimate $\Psi(\theta_0)$.

\subsection{Results for a deterministic number of terms} \label{section: simple case theory}

Following \cite{Chen2007,Shen1997}, our proofs of the validity of our data-adaptive series approach make heavy use of projection operators. We use $\pi_n := \pi_{n,\theta_n^0}$ to denote the projection operator for functions in $\Theta$ onto $\Theta_n=\Theta_{n,\theta_n^0}$ with respect to $\langle \cdot, \cdot \rangle$. For any function $\theta \in \Theta$, let $\Pi_{n,\theta}$ denote the operator that takes as input a function $g: \real^q \rightarrow \real^q$ for which $g \circ \theta \in L^2(P_0)$ and outputs a function $\Pi_{n,\theta}(g): \real^q \rightarrow \real^q$ such that $\Pi_{n,\theta}(g) \circ \theta=\pi_{n,\theta}(g \circ \theta)$. In other words, letting $\beta_j$ be the quantity that depends on $g$ and $\theta$ such that $\pi_{n,\theta}(g \circ \theta)=( \sum_{j=1}^{K} \beta_j \phi_j ) \circ \theta$, we define $\Pi_{n,\theta} (g) := \sum_{j=1}^K \beta_j \phi_j$. The operator $\Pi_{n,\theta}$ may also be interpreted as follows: letting $P_\theta$ be the distribution of $\theta(X)$ with $V=(X,Z) \sim P_0$, then $\Pi_{n,\theta}$ is the projection operator of functions $\real^q \rightarrow \real^q$ with respect to the $L^2(P_\theta)$-inner product. We use $\id$ to denote the identity function in $\real^q$.

We now present additional conditions we will require to ensure that $\Psi(\theta_n^*)$ is an efficient estimator of $\Psi(\theta_0)$.

\begin{conditionC}[Sufficient convergence rate of initial ML fit]
    \label{Ainit}
    $\| \theta_n^0 - \theta_0 \|=o_p(n^{-1/4})$.
\end{conditionC}

\begin{conditionC}[Sufficiently small estimation error]
    \label{Aestimation}
    $\| \theta_n^* - \pi_n(\theta_0) \|=o_p(n^{-1/4})$.
\end{conditionC}

\begin{conditionC}[Sufficiently small approximation error to $\id$ for $\Theta_{n,\theta_0}$]
    \label{Aapproxidentity}
    $\| \theta_0 - \Pi_{n,\theta_0}(\id) \circ \theta_0 \|=o(n^{-1/4})$.
\end{conditionC}

\begin{conditionC}[Sufficiently small approximation error to $\dot{\psi}$ for $\Theta_{n,\theta_0}$ and convergence rate of $\theta_n^*$]
    \label{Aapproxw}
    $\| [\dot{\psi} - \Pi_{n,\theta_0}(\dot{\psi})] \circ \theta_0 \| \cdot \| \theta_n^* - \theta_0 \|=o_p(n^{-1/2})$.
\end{conditionC}

Appendix~\ref{section: series additional regularity conditions} contains further technical conditions and Appendix~\ref{section: condition discussion} discusses their plausibility. As discussed in Appendix~\ref{section: condition discussion}, Conditions~\ref{Aestimation}--\ref{Aapproxw} typically imply restrictions on the growth rate of $K$: if $K$ grows too fast with $n$, then Condition~\ref{Aestimation} may be violated; if $K$ instead grows too slow, then Conditions~\ref{Aapproxidentity} and \ref{Aapproxw} may be violated. For the generalized moment $\Psi: \theta \mapsto P_0(f \circ \theta)$ with a fixed known function $f$ in Example~\ref{example: moments}, Condition~\ref{Aapproxw} typically also imposes a smoothness condition $f$ so that $f'$ can be approximated by the series well. Our conditions are closely related to the conditions in Theorem~1 of \cite{Shen1997}. Conditions~\ref{Ainit}--\ref{Aapproxidentity} and \ref{ALipschitzidentity} serve as sufficient conditions for the condition on the smoothness of $\Psi$ and the convergence rate of $\theta^*_n$ in Theorem~1 of \cite{Shen1997}. Together with Conditions~\ref{Aapproxw} and \ref{ALipschitzw}, we can derive Lemma~\ref{Lapproxw}, which is similar to the first part of Condition~C of \cite{Shen1997}. The empirical process condition \ref{Aempiricalprocess} is sufficient for Conditions~A, D and the second part of C in Theorem~1 in \cite{Shen1997}.

We now present a theorem ensuring the asymptotic linearity and efficiency of the plug-in estimator based on $\theta_n^*$.

\begin{theorem}[Efficiency of plug-in estimator]
    \label{Tefficiency}
    Under Conditions~\ref{Adloss}--\ref{APsiremainder} and \ref{Ainit}--\ref{Afinitevar}, $\Psi(\theta_n^*)$ is an asymptotically linear estimator of $\Psi(\theta_0)$ with the influence function being $v \mapsto \alpha_{0,\ell}^{-1} \{-\ell_{0}'[\dot{\psi} \circ \theta_0](v) + \expect_{P_0} [ \ell_{0}'[\dot{\psi} \circ \theta_0](V) ]\}$, that is,
    $$\Psi(\theta_n^*) = \Psi(\theta_0) + \frac{1}{n} \sum_{i=1}^n \alpha_{0,\ell}^{-1} \left\{ -\ell_{0}'[\dot{\psi} \circ \theta_0](V_i) + \expect_{P_0} \left[ \ell_{0}'[\dot{\psi} \circ \theta_0](V) \right] \right\}+o_p(n^{-1/2}).$$
    As a consequence, $\sqrt{n} [\Psi(\theta_n^*)-\Psi(\theta_0)] \overset{d}{\rightarrow} N(0, \xi^2)$ with $\xi^2 := \Var_{P_0}(\ell_{0}'[\dot{\psi} \circ \theta_0](V))/\alpha_{0,\ell}^2$. In addition, under Conditions~\ref{Acloseminimizer} and \ref{Aempiricalprocesslike} in Appendix~\ref{section: regularity additional conditions}, $\Psi(\hat{\theta}_n)$ is efficient under a nonparametric model.
\end{theorem}

\begin{remark} \label{remark: targeted series}
    Consider the general case in which it may not be true that $\dot{\Psi}$ can be represented as $\dot{\psi} \circ \theta_0$ for some $\dot{\psi}: \real^q \rightarrow \real^q$. If the analytic expression of $\dot{\Psi}$ can be derived and $\dot{\Psi}$ can be estimated by $\dot{\Psi}_n$ such that $\| \dot{\Psi}_n - \dot{\Psi} \| \cdot \| \theta_n^0 - \theta_0 \|=o_p(n^{-1/2})$, then our data-adaptive series can take a special form that is targeted towards $\Psi$. Specifically, letting $\vartheta_0:=(\theta_0,\dot{\Psi})^\top$ and $\varPsi(\vartheta_0):=\Psi(\theta_0)$, it is straightforward to show that the gradient of $\varPsi$ is $\dot{\varPsi}=(\dot{\Psi},0)^\top=(e_2,\bm{0})^\top \vartheta_0$ with $\bm{0}=(0,0)^\top$ and $e_2=(0,1)^\top$, which is a function composed with $\vartheta_0$. We can set $\vartheta_n^0=(\theta_n^0,\dot{\Psi}_n)^\top$ and $\Theta_n=\text{Span}\{\theta_n^0,\dot{\Psi}_n\}$ in our data-adaptive series. This approach does not have a growing number of terms in $\Theta_n$ and is not similar to sieve estimation, but can be treated as a special case of data-adaptive series. It can be shown that Conditions~\ref{Ainit}--\ref{Aapproxw} are still satisfied for $\vartheta$ and $\varPsi$ with this choice of $\Theta_n$, and hence our data-adaptive series estimator leads to an efficient plug-in estimator. We remark that the introduction of $\vartheta$ and $\varPsi$ is a purely theoretical device, and this targeted approach to estimation is quite similar to that used in the context of TMLE \citep{VanderLaan2006,VanderLaan2018}.
\end{remark}

\subsection{Summaries involving the marginal distribution of $X$} \label{section: generalized simple case}

We now generalize the setting considered thus far by allowing the parameter to depend both on $\theta_0$ and on $P_0$, i.e., estimating $\Psi(\theta_0, P_0)$. The example given at the beginning of Section~\ref{section: simple case method}, namely that of estimating $\Psi(\theta_0)=P_0 (f \circ \theta_0)$, is a special case of this more general setting. In what follows, we will make use of the following conditions:

\begin{conditionD}[Conditions with $P_0$ fixed] \label{Apackage}
    When we regard $\Psi(\theta_0, P_0)$ as the mapping $\theta \mapsto \Psi(\theta,P_0)$ evaluated at $\theta_0$, Conditions~\ref{Adloss}--\ref{APsiremainder}, \ref{Ainit}--\ref{Aapproxw} and \ref{ALipschitzidentity}--\ref{Afinitevar} are satisfied for estimating $\Psi(\theta_0, P_0)$.
\end{conditionD}

\begin{conditionD}[Hadamard differentiability with $\theta_0$ fixed] \label{AHadamard}
    The mapping $P \mapsto \Psi(\theta_0,P)$ is Hadamard differentiable at $P_0$.
\end{conditionD}

By the functional delta method, it follows that $\Psi(\theta_0,P_n)=\Psi(\theta_0,P_0) + P_n \text{IF}_0 + o_p(n^{-1/2})$ for a function $\text{IF}_0$ satisfying $P_0 \text{IF}_0=0$ and $P_0 \text{IF}_0^2 < \infty$.

\begin{conditionD}[Negligible second-order difference] \label{Asecondorderdifference}
    $$[\Psi(\theta_n^*,P_n) - \Psi(\theta_0,P_n)] - [\Psi(\theta_n^*,P_0) - \Psi(\theta_0,P_0)] = o_p(n^{-1/2}).$$
\end{conditionD}

This condition usually holds, for example, when $\Psi(\theta_0, P_0)=P_0 (f \circ \theta_0)$, as in this case the left-hand side is equal to $(P_n-P_0) (f \circ \theta_n^* - f \circ \theta_0)$, which is $o_p(n^{-1/2})$ under empirical process conditions.

\begin{theorem}[Asymptotic linearity of plug-in estimator]
    \label{TefficiencyP}
    Under Conditions~\ref{Apackage}--\ref{Asecondorderdifference}, $\Psi(\theta_n^*,P_n)$ is an asymptotically linear estimator of $\Psi(\theta_0,P_0)$ with influence function
    $$v \mapsto \alpha_{0,\ell}^{-1} \left\{-\ell_{0}'[\dot{\psi} \circ \theta_0](v) + \expect_{P_0}[\ell_{0}'[\dot{\psi} \circ \theta_0](V)] \right\} + \mathrm{IF}_0(V),$$
    that is,
    \begin{align*}
    \Psi(\theta_n^*,P_n) &= \Psi(\theta_0,P_0) + \frac{1}{n} \sum_{i=1}^n \left\{ -\alpha_{0,\ell}^{-1} \ell_{0}'[\dot{\psi} \circ \theta_0](V_i) + \alpha_{0,\ell}^{-1} \expect_{P_0} \left[ \ell_{0}'[\dot{\psi} \circ \theta_0](V) \right] + \mathrm{IF}(V_i) \right\} \\
    &\quad+o_p(n^{-1/2}).
    \end{align*}
    As a consequence, $\sqrt{n} [\Psi(\theta_n^*,P_n)-\Psi(\theta_0,P_0)] \overset{d}{\rightarrow} N(0, \xi^2)$ with $\xi^2 := \Var_{P_0}(\alpha_{0,\ell}^{-1} \ell_{0}'[\dot{\psi} \circ \theta_0](V) + \mathrm{IF}(V))$.
\end{theorem}

This result is easy to verify by decomposing $\Psi(\theta_n^*,P_n) - \Psi(\theta_0,P_0)$ as
\begin{align*}
& [\Psi(\theta_n^*,P_0) - \Psi(\theta_0,P_0)] + [\Psi(\theta_0,P_n) - \Psi(\theta_0,P_0)] \\
&\quad\quad\quad\quad+ \{[\Psi(\theta_n^*,P_n) - \Psi(\theta_0,P_n)] - [\Psi(\theta_n^*,P_0) - \Psi(\theta_0,P_0)]\}
\end{align*}
Moreover, under conditions similar to the conditions \ref{Acloseminimizer} and \ref{Aempiricalprocesslike} given in Appendix~\ref{section: regularity additional conditions}, we can show that $\Psi(\theta_n^*,P_n)$ is efficient under a nonparametric model.

\begin{remark}
    Conditions~\ref{AHadamard} and \ref{Asecondorderdifference} can be relaxed. Specifically, if $\hat{P}_n$ is an estimator of $P_0$ that satisfies that $\Psi(\theta_0,\hat{P}_n)=\Psi(\theta_0,P_0) + P_n \text{IF}_0 + o_p(n^{-1/2})$ for an influence function $\text{IF}_0$ and Condition~\ref{Asecondorderdifference} holds with $P_n$ replaced by $\hat{P}_n$, then $\Psi(\theta_n^*,\hat{P}_n)$ is an asymptotically linear estimator of $\Psi(\theta_0,P_0)$.
\end{remark}

\subsection{CV selection of the number of terms in data-adaptive series} \label{section: CV}

In the preceding subsections, we established the efficiency of the plug-in estimator based on suitable rates of growth for $K$ relative to the sample size $n$. In this subsection, we show that, under some conditions, such a $K$ can be selected by $k$-fold CV: after obtaining $\theta_n^0$, for each $K$ in a range of candidates, we can calculate the cross-validated risk from $k$ folds and choose the value of $K$ with the smallest CV risk. We denote the number of terms in the series that CV selects by $K^*$. In this section, we use $K$ in the subscripts for notation related to data-adaptive sieves-like spaces and projections; this represents a slight abuse of notation because, in Sections~\ref{section: simple case method} and \ref{section: simple case theory}, these subscripts were instead used for sample size $n$. That is, we use $\Theta_{K,\theta}$ to denote $\text{Span}\{ \phi_1, \phi_2, \ldots, \phi_K \} \circ \theta$, $\pi_{K,\theta}$ to denote the projection onto $\Theta_{K,\theta}$, $\Pi_{K,\theta}$ to denote the operator such that $\Pi_{K,\theta} (g) \circ \theta=\pi_{K,\theta}(g \circ \theta)$ for all $g: \real^q \rightarrow \real^q$ with $g \circ \theta \in L^2(P_0)$, and $\theta_n^\sharp := \theta_{K^*}^*(\theta_n^0)$ to be the data-adaptive series estimator based on $\theta_n^0 $ and $K^*$.

\begin{conditionC}[Bounded approximation error of $\dot{\psi}$ relative to $\id$]
    \label{Abadseries}
    There exists a constant $C>0$ such that, with probability tending to one, $\| \dot{\psi} \circ \theta_n^0 - \Pi_{K,\theta_n^0}(\dot{\psi}) \circ \theta_n^0 \| \leq C \| \theta_n^0 - \Pi_{K,\theta_n^0}(\id) \circ \theta_n^0 \|$ for all $K$.
\end{conditionC}

This condition is equivalent to
$$\| \dot{\psi} - \Pi_{K,\theta_n^0}(\dot{\psi}) \|_{L^2(P_{\theta_n^0})} \leq C \| \id - \Pi_{K,\theta_n^0}(\id) \|_{L^2(P_{\theta_n^0})}$$
for all $K$ with probability tending to one, which may be interpreted in terms of two simultaneous requirements. The first requirement is that the identity function $\id$ is not exactly contained in the span of $\phi_1,\ldots,\phi_K$ for any $K$, since otherwise, the right-hand side would be zero for all sufficiently large $K$. Therefore, common series such as polynomial and spline series are not permitted for general summaries. In contrast, other series such as trigonometric series and wavelets satisfy this requirement. The second requirement is that the approximation error of the chosen series for the identity function $\id$ is not much larger than $\dot{\psi}$. If a trigonometric or wavelet series is used, then this condition imposes a strong smoothness condition on derivatives of $\dot{\psi}$. Nonetheless, this may not be stringent in some interesting examples. For example, if $\Psi(\theta)=P_0 (f \circ \theta)$ for a fixed function $f$ in Example~\ref{example: moments}, then $\dot{\psi}$ equals $f'$ and hence can be expected to satisfy this strong smoothness condition provided that $f$ is infinitely differentiable with bounded derivatives. The estimands encountered in many applications involve $f$ satisfying this smoothness condition.

The following theorem justifies the use of $k$-fold CV to select $K$ under appropriate conditions.

\begin{theorem}[Efficiency of CV-based plug-in estimator]
    \label{Tcv}
    Assume that Conditions~\ref{Adloss}--\ref{APsiremainder}, \ref{Ainit}--\ref{Aapproxidentity}, \ref{Abadseries}, \ref{Aempiricalprocess} and \ref{Afinitevar} hold for a deterministic $K=K(n)$. Suppose part~\ref{ALipschitzw first half} of Condition~\ref{ALipschitzw} holds, then, with $\theta_n^\sharp:=\theta_{K^*}(\theta_n^0)$, $\Psi(\theta_n^\sharp)$ is an asymptotically linear estimator of $\Psi(\theta_0)$ with influence function $v \mapsto \alpha_{0,\ell}^{-1} \{-\ell_{0}'[\dot{\psi} \circ \theta_0](v) + \expect_{P_0} [ \ell_{0}'[\dot{\psi} \circ \theta_0](V) ]\}$, that is,
    $$\Psi(\theta_n^\sharp) = \Psi(\theta_0) + \frac{1}{n} \sum_{i=1}^n \alpha_{0,\ell}^{-1} \left\{ -\ell_{0}'[\dot{\psi} \circ \theta_0](V_i) + \expect_{P_0} \left[ \ell_{0}'[\dot{\psi} \circ \theta_0](V) \right] \right\}+o_p(n^{-1/2}).$$
    As a consequence, $\sqrt{n} [\Psi(\theta_n^\sharp)-\Psi(\theta_0)] \overset{d}{\rightarrow} N(0, \xi^2)$ with $\xi^2 := \Var_{P_0}(\ell_{0}'[\dot{\psi} \circ \theta_0](V))/\alpha_{0,\ell}^2$. In addition, under Conditions~\ref{Acloseminimizer} and \ref{Aempiricalprocesslike} in Appendix~\ref{section: regularity additional conditions}, $\Psi(\hat{\theta}_n)$ is efficient under a nonparametric model.
\end{theorem}

\subsection{Simulation} \label{section: sim}

\subsubsection{Demonstration of Theorem~\ref{Tcv}} \label{section: demonstrate data-adaptive series}

We illustrate our method in a simulation in which we take $\theta_0(x)=\expect_{P_0}[Z|X=x]$ and $\Psi(\theta_0)=P_0 \theta_0^2$. This is a special case of Example~\ref{example: moments}. The true function $\theta_0$ is chosen to be discontinuous, which violates the smoothness assumptions commonly required in traditional series estimation. In this case, $\dot{\psi}=2 \id$ and so the constant in Condition~\ref{Abadseries} is 2. We compare the performance of plug-in estimators based on three different nonparametric regressions: (i) polynomial regression with degree selected by 10-fold CV (poly), which results in a traditional sieve estimator, (ii) gradient boosting (xgb) \citep{Friedman2001,Friedman2002,Mason1999,Mason2000}, and (iii) data-adaptive trigonometric series estimation with gradient boosting as the initial ML fit and 10-fold CV to select the number of terms in the series (xgb.trig). We also compare these plug-in estimators with the one-step correction estimator \citep{Pfanzagl1982} based on gradient boosting (xgb.1step). Further details of this simulation can be found in Appendix~\ref{appendix: simulation}.

Fig~\ref{Fmse_bias} presents $n \cdot \text{MSE}$ and $\sqrt{n} \cdot |\text{bias}|$ for each estimator, whereas Table \ref{TableCI} presents the coverage probability of 95\% Wald CIs based on these estimators. We find that xgb.trig and xgb.1step estimators perform well, while poly and xgb plug-in estimators do not appear to be efficient. Since polynomial series estimators only work well when estimating smooth functions, in this simulation, we would not expect the fit from the polynomial series estimator to converge sufficiently fast, and consequently, we would not expect the resulting plug-in estimator to be efficient. In contrast, gradient boosting is a flexible ML method that can learn discontinuous functions, so we can expect an efficient plug-in estimator based on this ML method. However, gradient boosting is not designed to approximately solve the estimating equation that achieves the small-bias property for this particular summary, so we would not expect its na\"{i}ve plug-in estimator to be efficient. Based on gradient boosting, our estimator and the one-step corrected estimator both appear to be efficient, but our method has the advantage of being a plug-in estimator. Moreover, the construction of our estimator does not require knowledge of the analytic expression of an influence function.

\begin{figure}[ht!]
    \centering
    \includegraphics[scale=0.65]{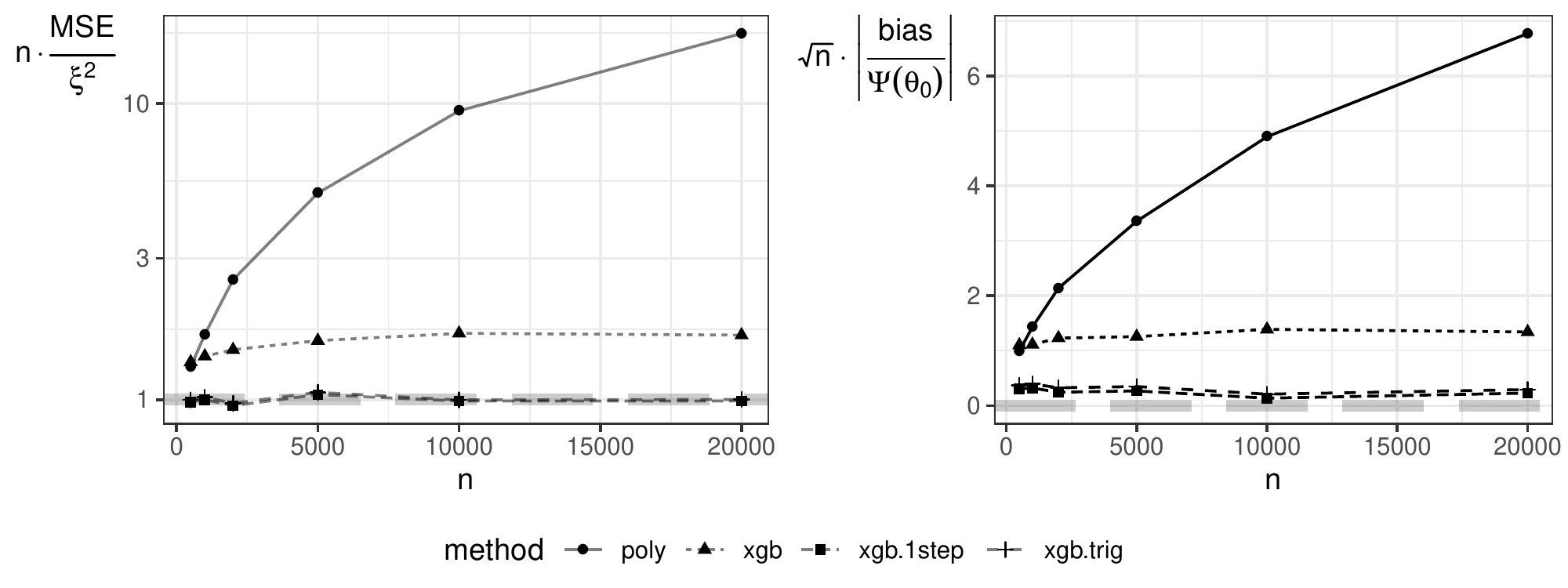}
    \caption{The relative MSE, $n \cdot \text{MSE}/\xi^2$, and the relative absolute bias, $\sqrt{n} \cdot |\text{bias}/\Psi(\theta_0)|$, of estimators of $\Psi(\theta_0)=P_0 \theta_0^2$. $\xi^2 := P_0 \text{IF}^2$ is the asymptotic variance that the $n \cdot \text{MSE}$ of an AL estimator should converge to. poly: plug-in estimator based on polynomial sieve estimation. xgb: plug-in estimator based on gradient boosting. xgb.1step: one-step correction (debiasing) of the plug-in estimator based on gradient boosting. xgb.trig: data-adaptive series with trigonometric series composed with gradient boosting. All tuning parameters are CV-selected. The y-axis for relative MSE is scaled based on logarithm for readability. Note that the $n \cdot \text{MSE}$ for xgb.trig and xgb.1step tend to $\xi^2$, but those for poly and xgb do not.}
    \label{Fmse_bias}
\end{figure}

\begin{table}[ht!]
    \centering
    \caption{Coverage probability of 95\% Wald CI based on estimators of $\Psi(\theta_0)=P_0 \theta_0^2$. poly: plug-in estimator based on polynomial sieve estimation. xgb: plug-in estimator based on gradient boosting. xgb.1step: one-step correction (debiasing) of the plug-in estimator based on gradient boosting. xgb.trig: data-adaptive series with trigonometric series composed with gradient boosting. All tuning parameters are CV-selected. The CI is constructed based on the influence function. The coverage probabilities for xgb.trig and xgb.1step are approximately 95\%, but those for poly and xgb are not.}
    \label{TableCI}
    \begin{tabular}{rrrrr}
        \hline
        n & poly & xgb & xgb.1step & xgb.trig \\ 
        \hline
        500 & 0.90 & 0.90 & 0.95 & 0.95 \\ 
        1000 & 0.86 & 0.89 & 0.95 & 0.95 \\ 
        2000 & 0.74 & 0.88 & 0.96 & 0.96 \\ 
        5000 & 0.47 & 0.88 & 0.94 & 0.94 \\ 
        10000 & 0.16 & 0.87 & 0.95 & 0.96 \\ 
        20000 & 0.02 & 0.86 & 0.96 & 0.96 \\
        \hline
    \end{tabular}
\end{table}

We also investigate the effect of the choice of $K$ on the performance of our method. Fig~\ref{FmseK} presents $n \cdot \text{MSE}$ for the data-adaptive series estimator with different choices of $K$. We can see that our method is insensitive to the choice of $K$ in this simulation setting. Although a relatively small $K$ performs better, choosing a much larger $K$ does not appear to substantially harm the behavior of the estimator. This insensitivity to the selected tuning parameter suggests that in some applications, without using CV, an almost arbitrary choice of $K$ that is sufficiently large might perform well.

\begin{figure}[ht!]
    \centering
    \includegraphics[scale=0.65]{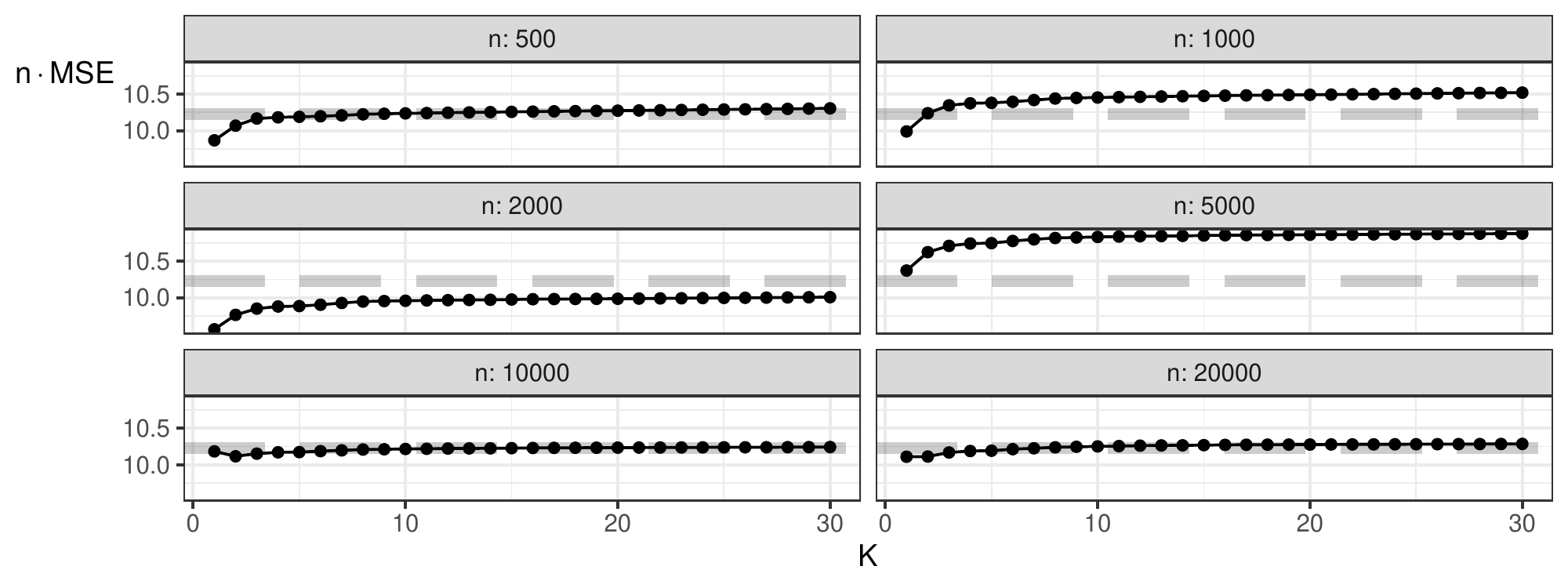}
    \caption{$n \cdot \text{MSE}$ of estimators of $\Psi(\theta_0)=P_0 \theta_0^2$ based on data-adaptive series with different choices of $K$. The horizontal gray thick dashed line is the asymptotic variance that the $n \cdot \text{MSE}$ of an AL estimator should converge to, $\xi^2 := P_0 \text{IF}^2$. Note that $n \cdot \text{MSE}$ is not sensitive to the choice of $K$ over a wide range of $K$.}
    \label{FmseK}
\end{figure}

\subsubsection{Violation of Condition~\ref{Abadseries}} \label{section: rough Psi}

For the $k$-fold CV selection of $K$ in our method to yield an efficient plug-in estimator, $\Psi$ must be highly smooth in the sense that $\dot{\psi}$ can be approximated by the series about as well as can the identity function (see Condition~\ref{Abadseries}). Although we have argued that this condition is reasonable, in this section, we explore via simulation the behavior of our method based on CV when $\dot{\psi}$ is rough. We again take $\theta_0: x \mapsto \expect_{P_0}[Z|X=x]$ and an artificial summary $\Psi(\theta_0)=P_0 (f \circ \theta_0)$, where $f$ is an element of $C^1[-1,1]$ but not of $C^2[-1,1]$. In this case, $\dot{\psi}=f'$ is very rough, so we do not expect it to be approximated by a trigonometric series as well as the identity function. However, it is sufficiently smooth to allow for the existence of a deterministic $K$ that achieves efficiency. Further simulation details are provided in Appendix~\ref{appendix: simulation}.

Table~\ref{Tablenonsmooth} presents the performance of our estimator based on 10-fold CV. We note that it performs reasonably well in terms of the $n\cdot \text{MSE}$ criterion. However, it is unclear whether its scaled bias converges to zero for large $n$, so our method may be too biased. The coverage of 95\% Wald CIs is close to the nominal level, suggesting that the bias is fairly small relative to the standard error of the estimator at the sample sizes considered. One possible explanation for the good performance observed is that the $L^2(P_0)$-convergence rate of $\theta_n^*$ is much faster than $n^{-1/4}$, which allows for a slower convergence rate of the approximation error $\| \dot{\psi} \circ \theta_0 - \Pi_{n,\theta_0}(\dot{\psi}) \circ \theta_0 \|$ (see Appendix~\ref{section: condition discussion}). This simulation shows that our proposed method may still perform well even if Condition~\ref{Abadseries} is violated, especially when the initial ML fit is close to the unknown function.

\begin{table}[ht!]
    \centering
    \caption{Performance of the plug-in estimator of $\Psi(\theta_0)=P_0 (f \circ \theta_0)$ based on data-adaptive series. Here $f$ is not infinitely differentiable. The relative MSE is $n \cdot \text{MSE}/\xi^2$ where $\xi^2 := P_0 \text{IF}^2$ is the asymptotic variance that the $n \cdot \text{MSE}$ of an AL estimator should converge to; the root-$n$ abs relative bias is $\sqrt{n} |\text{bias}/\Psi(\theta_0)|$. The performance appears to be acceptable in view of the small $\text{MSE}$ and reasonable CI coverage.}
    \label{Tablenonsmooth}
    \begin{tabular}{rrrr}
        \hline
        n & relative $\text{MSE}$ & root-$n$ absolute relative bias & 95\% Wald CI coverage \\ 
        \hline
        500 & 0.88 & 3.95 & 0.97 \\
        1000 & 0.89 & 3.73 & 0.96 \\
        2000 & 0.79 & 3.15 & 0.97 \\
        5000 & 0.78 & 2.02 & 0.97 \\
        10000 & 0.88 & 2.57 & 0.97 \\
        20000 & 0.88 & 1.75 & 0.96 \\
        \hline
    \end{tabular}
\end{table}

\section{Generalized data-adaptive series} \label{section: general case}

\subsection{Proposed method} \label{section: general case method}

As in Section~\ref{section: simple case}, we consider the case that $\Psi$ is scalar-valued in this section. The assumption that $\dot{\Psi}=\dot{\psi} \circ \theta_0$ may be too restrictive for general summaries as in Examples~\ref{example: average derivative}--\ref{example: treatment effect heterogeneity}, especially if $\dot{\Psi}$ is not derived analytically (see Remark~\ref{remark: targeted series}). In this section, we generalize the method in Section~\ref{section: simple case} to deal with these summaries. Letting $\id_x$ be the identity function defined on $\mathcal{X}$, we can readily generalize the above method to the case where $\dot{\Psi}$ can be represented as $\dot{\psi} \circ (\theta_0, \id_x)$ for a function $\dot{\psi}: \real^q \times \mathcal{X} \rightarrow \real^q$; that is, $\dot{\Psi}(x) = \dot{\psi}(\theta_0(x),x)$. This form holds trivially if we set $\dot{\psi}(t,x)=\dot{\Psi}(x)$, i.e., $\dot{\psi}$ is independent of its first argument, but we can utilize flexible ML methods if $\dot{\psi}$ is nontrivial. Again, we assume $\Theta$ is a vector space of $\real^q$-valued function equipped with the $L^2(P_0)$-inner product. We assume $\dot{\psi}$ can be approximated well by a basis $\phi_1,\phi_2,\ldots: \real^q \times \mathcal{X} \rightarrow \real^q$, and consider the data-adaptive sieve-like subspace $\Theta_n := \Theta_{n,\theta_n^0} := \text{Span}\{ \phi_1, \ldots, \phi_K \} \circ (\theta_n^0,\id_x)$. We propose to use $\Psi(\theta_n^*)$ to estimate $\Psi(\theta_0)$, where $\theta_n^*=\theta_n^*(\theta_n^0) \in \argmin_{\theta \in \Theta_n} n^{-1} \sum_{i=1}^{n} \ell(\theta)(V_i)$ denotes the series estimator within $\Theta_n$ minimizing the empirical risk.

\subsection{Results for proposed method} \label{section: general case theory}

With a slight abuse of notation, in this section we use $\id$ to denote the function $(t,x) \mapsto t$ where $t \in \real^q$ and $x \in \mathcal{X}$. Again, we use $\pi_n := \pi_{n,\theta_n^0}$ to denote the projection operator onto $\Theta_{n,\theta_n^0}$. Let $\Pi_{n,\theta}$ be defined such that, for any function $g: \real^q \times \mathcal{X} \rightarrow \real^q$ with $g \circ (\theta,\id_x) \in L^2(P_0)$, it holds that $\Pi_{n,\theta} (g) \circ (\theta,\id_x) =\pi_{n,\theta} (g \circ (\theta,\id_x))$; that is, letting $\beta_j$ be the quantity that depends on $g$ and $\theta$ such that $\pi_{n,\theta} (g \circ (\theta,\id_x))=( \sum_{j=1}^K \beta_j \phi_j ) \circ (\theta,\id_x)$, we define $\Pi_{n,\theta} (g) := \sum_{j=1}^K \beta_j \phi_j$.

We introduce conditions and derive theoretical results that are parallel to those in Section \ref{section: simple case}.

\begin{condition2}{Aapproxidentity}[Sufficiently small approximation  error to $\id$ for $\Theta_{n,\theta_0}$]
    \label{Aapproxidentity2}
    $\| \theta_0 - \Pi_{n,\theta_0}(\id) \circ (\theta_0,\id_x) \|=o(n^{-1/4})$.
\end{condition2}

\begin{condition2}{Aapproxw}[Sufficiently small approximation error to $\dot{\psi}$ for $\Theta_{n,\theta_0}$ and convergence rate of $\theta_n^*$]
    \label{Aapproxw2}
    $\| [\dot{\psi} - \Pi_{n,\theta_0}(\dot{\psi})] \circ (\theta_0,\id_x) \| \cdot \| \theta_n^*-\theta_0 \|=o_p(n^{-1/2})$.
\end{condition2}

Additional regularity conditions can be found in Appendix~\ref{section: general series additional regularity conditions}. Note that $\dot{\Psi}$ may depend on components of $P_0$ other than $\theta_0$, Condition~\ref{Aapproxw2} may impose smoothness conditions on these components so that $\dot{\psi}$ can be well approximated by the chosen series. For example, in Example~\ref{example: average derivative}, Condition~\ref{Aapproxw2} requires that $p_0'/p_0$ and the propensity score can be approximated by the series well; in Examples~\ref{example: ATE} and \ref{example: treatment effect heterogeneity}, Condition~\ref{Aapproxw2} imposes the same requirement on the propensity score. We now present a theorem that establishes the efficiency of the plug-in estimator based on $\theta_n^*$.

\begin{theorem}[Efficiency of plug-in estimator]
    \label{Tefficiency2}
    Under Conditions~\ref{Adloss}--\ref{APsiremainder}, \ref{Ainit}, \ref{Aestimation}, \ref{Aapproxidentity2}, \ref{Aapproxw2}, \ref{ALipschitzidentity2}, \ref{ALipschitzw2}, \ref{Aempiricalprocess} and \ref{Afinitevar}, $\Psi(\theta_n^*)$ is an asymptotically linear estimator of $\Psi(\theta_0)$ with influence function $v \mapsto \alpha_{0,\ell}^{-1} \{-\ell_{0}'[\dot{\psi} \circ (\theta_0, \id_x)](v) + \expect_{P_0} [ \ell_{0}'[\dot{\psi} \circ (\theta_0, \id_x)](V) ]\}$, that is,
    $$\Psi(\theta_n^*) = \Psi(\theta_0) + \frac{1}{n} \sum_{i=1}^n \alpha_{0,\ell}^{-1} \left\{ -\ell_{0}'[\dot{\psi} \circ (\theta_0, \id_x)](V_i) + \expect_{P_0} \left[ \ell_{0}'[\dot{\psi} \circ (\theta_0, \id_x)](V) \right] \right\}+o_p(n^{-1/2}).$$
    As a consequence, $\sqrt{n} [\Psi(\theta_n^*)-\Psi(\theta_0)] \overset{d}{\rightarrow} N(0, \xi^2)$ with $\xi^2 := \Var_{P_0}(\ell_{0}'[\dot{\psi} \circ (\theta_0, \id_x)](V))/\alpha_{0,\ell}^2$. In addition, under Conditions~\ref{Acloseminimizer} and \ref{Aempiricalprocesslike} in Appendix~\ref{section: regularity additional conditions}, $\Psi(\hat{\theta}_n)$ is efficient under nonparametric models.
\end{theorem}

\begin{remark}
    When $\Psi$ depends on both $\theta_0$ and $P_0$, we can readily adapt this method as in Section~\ref{section: generalized simple case}.
\end{remark}

We now present a condition for selecting $K$ via $k$-fold CV, in parallel with Condition~\ref{Abadseries} from Section~\ref{section: CV}.

\begin{condition2}{Abadseries}[Bounded approximation error of $\dot{\psi}$ relative to $\id$]
    \label{Abadseries2}
    There exists a constant $C>0$ such that, with probability tending to one, $\| \dot{\psi} \circ (\theta_n^0,\id_x) - \Pi_{K,\theta_n^0}(\dot{\psi}) \circ (\theta_n^0,\id_x) \| \leq C \| \theta_n^0 - \Pi_{K,\theta_n^0}(\id) \circ (\theta_n^0,\id_x) \|$ for all $K$.
\end{condition2}

\begin{remark}
    Similarly to Condition~\ref{Abadseries}, Condition~\ref{Abadseries2} requires that the identity function $\id$ is not contained in the span of finitely many terms of the chosen series and that $\dot{\psi}$ is sufficiently smooth so that $\dot{\psi}$ can be approximated well by the chosen series. However, Condition~\ref{Abadseries2} may be far more stringent than Condition~\ref{Abadseries}. In fact, it may be overly stringent in practice. Since $\dot{\Psi}$ may depend on components of $P_0$ other than $\theta_0$, Condition~\ref{Abadseries2} may require these components to be sufficiently smooth. When a common candidate series such as the trigonometric series is used, a sufficient condition for Condition~\ref{Abadseries2} is that $\dot{\psi}$ is infinitely differentiable with bounded derivatives, which further imposes assumptions on the smoothness of other components of $P_0$. For example, in Example~\ref{example: average derivative}, a sufficient condition for Condition~\ref{Abadseries2} is that $p_0'/p_0$ is infinitely differentiable with bounded derivatives; in Examples~\ref{example: ATE} and \ref{example: treatment effect heterogeneity}, a sufficient condition for Condition~\ref{Abadseries2} is that Condition~\ref{Abadseries2} is that the propensity score function satisfies the same requirement. Due to the stringency of Condition~\ref{Abadseries2}, we conduct a simulation in Section~\ref{section: sim2 violation of Abadseries2} to understand the performance of our proposed method when this condition is violated. The simulation appears to indicate that our method may be robust against violation of Condition~\ref{Abadseries2}.
\end{remark}

The following theorem shows that $k$-fold CV can be used to select $K$ under certain conditions.

\begin{theorem}[Efficiency of CV-based plug-in estimator]
    \label{Tcv2}
    Assume Conditions~\ref{Adloss}--\ref{APsiremainder},~\ref{Ainit}, \ref{Aestimation}, \ref{Aapproxidentity2}, \ref{ALipschitzidentity2}, \ref{ALipschitzw2}, \ref{Aempiricalprocess} and \ref{Afinitevar} hold for a deterministic $K=K(n)$. Suppose that part~\ref{ALipschitzw2 first half} of Condition~\ref{ALipschitzw2} holds. With $\theta_n^\sharp:=\theta_{K^*}(\theta_n^0)$, $\Psi(\theta_n^\sharp)$ is an asymptotically linear estimator of $\Psi(\theta_0)$ with influence function $v \mapsto \alpha_{0,\ell}^{-1} \{-\ell_{0}'[\dot{\psi} \circ (\theta_0,\id_x)](v) + \expect_{P_0} [ \ell_{0}'[\dot{\psi} \circ (\theta_0,\id_x)](V) ]\}$, that is,
    $$\Psi(\theta_n^\sharp) = \Psi(\theta_0) + \frac{1}{n} \sum_{i=1}^n \alpha_{0,\ell}^{-1} \left\{ -\ell_{0}'[\dot{\psi} \circ (\theta_0,\id_x)](V_i) + \expect_{P_0} \left[ \ell_{0}'[\dot{\psi} \circ (\theta_0,\id_x)](V) \right] \right\}+o_p(n^{-1/2}).$$
    Therefore, $\sqrt{n} [\Psi(\theta_n^\sharp)-\Psi(\theta_0)] \overset{d}{\rightarrow} N(0, \xi^2 )$ with $\xi^2 := \Var_{P_0}(\ell_{0}'[\dot{\psi} \circ \theta_0](V))/\alpha_{0,\ell}^2$. In addition, under Conditions~\ref{Acloseminimizer} and \ref{Aempiricalprocesslike} in Appendix~\ref{section: regularity additional conditions}, $\Psi(\hat{\theta}_n)$ is efficient under a nonparametric model.
\end{theorem}

\subsection{Simulation} \label{section: sim2}

In the following simulations, we consider the problem in Example~\ref{example: treatment effect heterogeneity}. As we show in Appendix~\ref{appendix: change norm}, letting $g_0: x \mapsto P_0 (A=1|X=x)$ be the propensity score and setting $\theta=(\mu_0,\mu_1)$, with $\ell(\theta): v \mapsto a [z-\mu_1(x)]^2 + (1-a) [z-\mu_0(x)]^2$, the generalized data-adaptive series methodology may be used to obtain an efficient estimator. As in Section~\ref{section: sim}, we conduct two simulation studies, the first demonstrating Theorem~\ref{Tcv2} and the other exploring the robustness of CV against violation of Condition~\ref{Abadseries2}.

\subsubsection{Demonstration of Theorem~\ref{Tcv2}}

We choose $\theta_0$ to be a discontinuous function while $g_0$ is highly smooth. We compare the performance of plug-in estimators based on three different nonparametric regressions: (i) polynomial regression with the degree selected by 5-fold CV (poly), which results in a traditional sieve estimator, (ii) gradient boosting (xgb) \citep{Friedman2001,Friedman2002,Mason1999,Mason2000}, and (iii) generalized data-adaptive trigonometric series estimation with gradient boosting as the initial ML fit and 5-fold CV to select the number of terms in the series (xgb.trig). Further details of the simulation setting are provided in Appendix~\ref{appendix: simulation}.

Fig~\ref{Fmse_bias2} presents $n \cdot \text{MSE}$ and $\sqrt{n} \cdot |\text{bias}|$ for each estimator, whereas Table \ref{TableCI2} presents the coverage probability of 95\% Wald CIs based on these estimators. There are a few runs in the simulation with noticeably poor behavior, so we trimmed the most extreme values when computing MSE and bias in Fig~\ref{Fmse_bias2} (1\% of all Monte Carlo runs). The outliers may be caused by the performance of gradient boosting and the instability of 5-fold CV. In practice, the user may ensemble more ML methods and use 10-fold CV to mitigate such behavior. We note that xgb.trig and xgb.1step estimators perform well, while poly and xgb plug-in estimators do not appear to be efficient. Based on gradient boosting, our estimator and the one-step corrected estimator both appear to be efficient, but the construction of our estimator has the advantage of not requiring the analytic expression of an influence function.

\begin{figure}[ht!]
    \centering
    \includegraphics[scale=0.65]{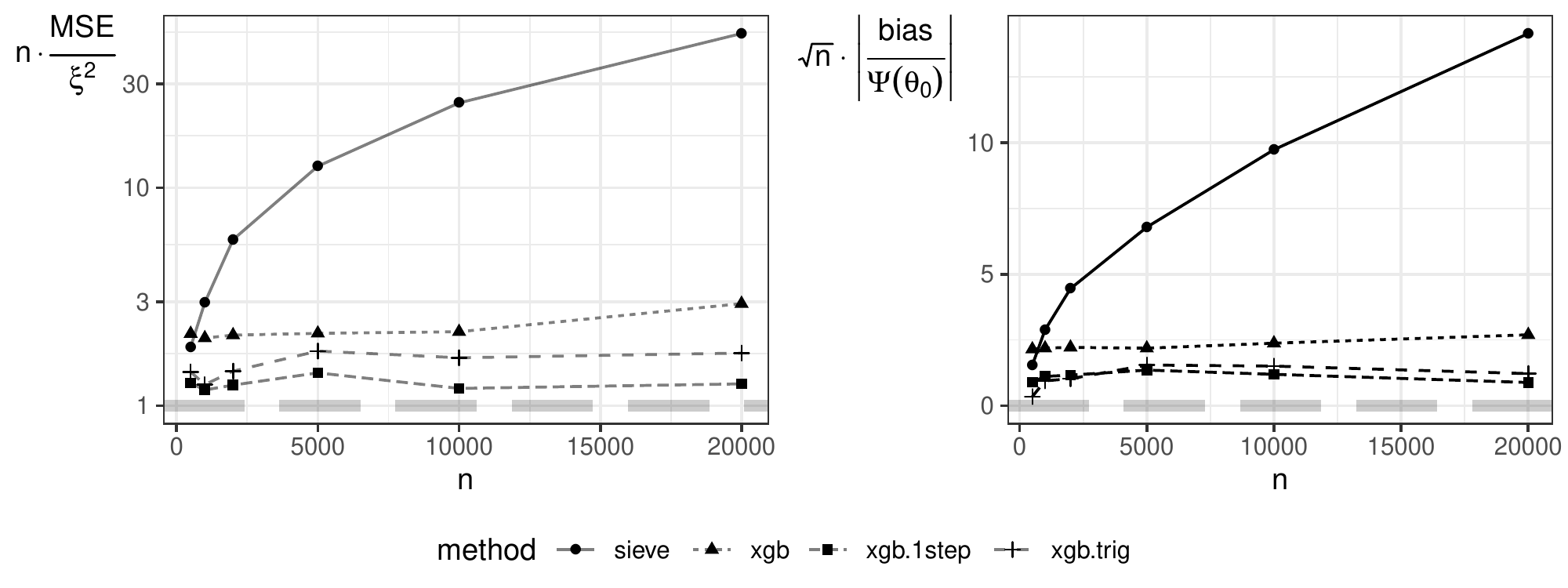}
    \caption{The relative MSE, $n \cdot \text{MSE}/\xi^2$, and the relative absolute bias, $\sqrt{n} \cdot |\text{bias}/\Psi(\theta_0)|$, of estimators of $\Psi(\theta_0)=\Var_{P_0}(\mu_{0,1}(X)-\mu_{0,0}(X))$ where $\mu_{0,a}: x \mapsto \expect_{P_0}[Y|A=a,X=x]$. $\xi^2 := P_0 \text{IF}^2$ is the asymptotic variance that the $n \cdot \text{MSE}$ of an AL estimator should converge to. poly: plug-in estimator based on polynomial sieve estimation. xgb: plug-in estimator based on gradient boosting. xgb.1step: one-step correction (debiasing) of the plug-in estimator based on gradient boosting. xgb.trig: data-adaptive series with trigonometric series composed with gradient boosting. All tuning parameters are CV-selected. The y-axis for relative MSE is scaled based on logarithm for readability. Note that the $n \cdot \text{MSE}$ for xgb.trig and xgb.1step tend to $\xi^2$, but those for poly and xgb do not.}
    \label{Fmse_bias2}
\end{figure}

\begin{table}[ht!]
    \centering
    \caption{Coverage probability of 95\% Wald CI based on estimators of $\Psi(\theta_0)=\Var_{P_0}(\mu_{0,1}(X)-\mu_{0,0}(X))$ where $\mu_{0,a}: x \mapsto \expect_{P_0}[Y|A=a,X=x]$. poly: plug-in estimator based on polynomial sieve estimation. xgb: plug-in estimator based on gradient boosting. xgb.1step: one-step correction (debiasing) of the plug-in estimator based on gradient boosting. xgb.trig: data-adaptive series with trigonometric series composed with gradient boosting. All tuning parameters are CV-selected. The CI is constructed based on the influence function. The coverage probabilities for xgb.trig and xgb.1step are relatively close to 95\%, but those for poly and xgb are not.}
    \label{TableCI2}
    \begin{tabular}{rrrrr}
        \hline
        n & poly & xgb & xgb.1step & xgb.trig \\ 
        \hline
        500 & 0.85 & 0.76 & 0.89 & 0.90 \\
        1000 & 0.68 & 0.78 & 0.93 & 0.93 \\
        2000 & 0.44 & 0.81 & 0.93 & 0.92 \\
        5000 & 0.11 & 0.80 & 0.89 & 0.87 \\
        10000 & 0.00 & 0.79 & 0.92 & 0.90 \\
        20000 & 0.00 & 0.67 & 0.91 & 0.88 \\
        \hline
    \end{tabular}
\end{table}

\subsubsection{Violation of Condition~\ref{Abadseries2}} \label{section: sim2 violation of Abadseries2}

We also study via simulation the behavior of our estimator when Condition~\ref{Abadseries2} is violated. We note that whether Condition~\ref{Abadseries2} holds depends on the smoothness of $g_0$. We choose $g_0$ to be rougher than $\id$ with $g_0$ being an element of $C^2[-1,1]$ but not of $C^3[-1,1]$. Consequently, $\dot{\Psi}$ cannot be approximated by our generalized data-adaptive series as well as $\id$, but its smoothness is sufficient for the existence of a deterministic $K$ to achieve efficiency. Appendix~\ref{appendix: simulation} describes further details of this simulation setting.

Table~\ref{Tablenonsmooth2} presents the performance of our estimator based on 5-fold CV. We observe that its scaled $\text{MSE}$ appears to converge to one, but it is unclear whether its scaled $\text{bias}$ converges to zero for large $n$, and so our method may be overly biased.. The coverage of 95\% Wald CIs is close to the nominal level, suggesting that the bias may be fairly small relative to the standard error of the estimator at the sample sizes considered. Therefore, according to this simulation, our generalized data-adaptive series methodology appears to be robust against violation of Condition~\ref{Abadseries2}.

\begin{table}[ht!]
    \centering
    \caption{Performance of the plug-in estimator of $\Psi(\theta_0)=\Var_{P_0}(\mu_{0,1}(X)-\mu_{0,0}(X))$ where $\mu_{0,a}: x \mapsto \expect_{P_0}[Z|A=a,X=x]$ based on data-adaptive series. Here the propensity score $g_0: x \mapsto \expect_{P_0}[A|X=x]$ is rough. The relative MSE is $n \cdot \text{MSE}/\xi^2$ where $\xi^2 := P_0 \text{IF}^2$ is the asymptotic variance that the $n \cdot \text{MSE}$ of an AL estimator should converge to; the root-$n$ abs relative bias is $\sqrt{n} |\text{bias}/\Psi(\theta_0)|$. The performance appears to be acceptable in view of small $\text{MSE}$ and reasonable CI coverage.}
    \label{Tablenonsmooth2}
    \begin{tabular}{rrrr}
        \hline
        n & relative $\text{MSE}$ & root-$n$ absolute relative bias & 95\% Wald CI coverage \\ 
        \hline
        500 & 1.02 & 0.28 & 0.92 \\ 
        1000 & 1.13 & 0.26 & 0.91 \\ 
        2000 & 1.10 & 0.19 & 0.94 \\ 
        5000 & 1.03 & 0.02 & 0.93 \\ 
        10000 & 0.96 & 0.23 & 0.95 \\ 
        20000 & 0.99 & 0.24 & 0.94 \\ 
        \hline
    \end{tabular}
\end{table}

\section{Discussion} \label{section: discussion}

Numerous methods have been proposed to construct efficient estimators for statistical parameters under a nonparametric model, but each of them has one or more of the following undesirable limitations: (i) their construction may require specialized expertise that is not accessible to most statisticians; (ii) for any given data set, there may be little guidance, if any, on how to select a key tuning parameter; and (iii) they may require stringent smoothness conditions, especially on derivatives. In this paper, we propose two sieve-like methods that can partially overcome these difficulties.

Our first approach, namely that based on HAL, can be further generalized to the case in which the flexible fit is an empirical risk minimizer over a function class assumed to contain the unknown function. The key condition \ref{AM} may be modified in that case as long as it ensures that certain perturbations of the unknown function still lie in that function class. We note that our methods may also be applied under semiparametric models.

A major direction for future work is to construct valid CIs without the knowledge of the influence function of the resulting plug-in estimator. The nonparametric bootstrap is in general invalid when the overall summary is not Hadamard differentiable and especially when the method relies on CV \citep{bickel1997,hall2013}, but a model-based bootstrap is a possible solution (Chapter~28 of \cite{VanderLaan2018}). In many cases only certain components of the true data-generating distribution must be estimated to obtain a plug-in estimator, while its variance may depend on other components that are not explicitly estimated. Therefore, generating valid model-based bootstrap samples is generally difficult.

Our proposed sieve-like methods may be used to construct efficient plug-in estimators for new applications in which the relevant theoretical results are difficult to derive. They may also inspire new methods to construct such estimators under weaker conditions.

\begin{appendices}

\section{Modification of chosen norm for evaluating the conditions: case study of mean counterfactual outcome} \label{appendix: change norm}

In this appendix, we consider a parameter that requires a modification in the chosen norm for evaluating the conditions. In particular, we discuss estimating counterfactual mean outcome in Example~\ref{example: ATE}.

Let $g_0: x \mapsto P_0(A=1|X=x)$ be the propensity score function. A natural choice of the loss function is $\ell(\theta): v \mapsto a[y-\theta(x)]^2$. Indeed, learning a function with this loss function is equivalent to fitting a function within the stratum of observations that received treatment 1. Unfortunately, this loss function does not satisfy Condition~\ref{Aquadraticloss} with $L^2(P_0)$-norm, because $P_0 \{\ell(\theta) - \ell(\theta_0)\} = P_0 \{g_0 \cdot (\theta-\theta_0)^2\}$ cannot be well approximated by $\alpha_{0,\ell} P_0 \{(\theta-\theta_0)^2\}/2$ for any constant $\alpha_{0,\ell}>0$ unless $g_0$ is a constant. One way to overcome this challenge is to choose the alternative inner product $\langle \theta_1, \theta_2 \rangle_{g_0} := P_0 \{g_0 \theta_1 \theta_2\}$ and its induced norm $\| \cdot \|_{g_0}$. In this case, Condition~\ref{Aquadraticloss} is satisfied once $\| \cdot \|$ is replaced by $\| \cdot \|_{g_0}$ in the condition statement. Under this choice, $\Psi'_{\theta_0}=P_0 (\theta-\theta_0)=\langle 1/g_0, \theta-\theta_0 \rangle_{g_0}$. We may redefine the corresponding $\dot{\Psi}$ similarly as the function that satisfies
$$\Psi'_{\theta_0} = \langle \dot{\Psi}, \theta-\theta_0 \rangle_{g_0},$$
and it immediately follows that $\dot{\Psi}=1/g_0$. Moreover, under a strong positivity condition, namely $g_0(X) \geq \delta_g >0$ a.s. for some $\delta_g$, which is a typical condition in causal inference literature \citep{VanderLaan2018,Yang2018}, then it is straightforward to show that $\delta_g \| \cdot \| \leq \| \cdot \|_{g_0} \leq \| \cdot \|$; that is, $\| \cdot \|_{g_0}$ is equivalent to $L^2(P_0)$-norm. Using this fact, it can be shown that all other conditions with respect to the $L^2(P_0)$-inner product are equivalent to the corresponding conditions with respect to $\langle \cdot,\cdot \rangle_{g_0}$.

Therefore, the data-adaptive series can be applied to estimation of the counterfactual mean outcome under our conditions for $L^2(P_0)$-inner product. If we use the targeted form in Remark~\ref{remark: targeted series}, then we need a flexible estimator of $g_0$ and the procedure is almost identical to a TMLE \citep{VanderLaan2018}. If we use the generalized data-adaptive series, we would require sufficient amount of smoothness for $g_0(\cdot)$. In the latter case, the change in norm when evaluating the conditions is a purely technical device and the estimation procedure is the same as would have been used if we had used the $L^2(P_0)$-norm. We also note that the same argument may be used to show that in Example~\ref{example: treatment effect heterogeneity}, with $\ell(\theta): v \mapsto a[z-\mu_1(x)]^2 + (1-a)[z-\mu_0(x)]^2$ being the usual squared-error loss, we may choose the alternative inner product $\langle \theta_1, \theta_2 \rangle_{g_0} := P_0 \{\theta_1^\top \cdot \mathrm{diag}(1-g_0,g_0) \cdot \theta_2\}$ and find that $\dot{\Psi} = (-2/(1-g_0) \cdot [(\mu_{01}-\mu_{00}) - P_0(\mu_{01}-\mu_{00})], 2/g_0 \cdot [(\mu_{01}-\mu_{00}) - P_0(\mu_{01}-\mu_{00})])^\top$, as we did in Section~\ref{section: sim2}.

\section{Additional conditions} \label{section: additional technical conditions}

Throughout the rest of this appendix, we use $C$ to denote a general absolute positive constant that can vary line by line.

\subsection{HAL} \label{section: HAL additional regularity conditions}

\begin{conditionB}[Empirical processes conditions]
    \label{AHALempiricalprocess}
    For any fixed $\vartheta \in \Theta_{\vnorm,M}$ and some $\Delta > 0$, it holds that $\ell(\theta)$, $\ell_{0}'[\theta-\theta_0]$ and $\{r[\theta-\theta_0] - r[\theta + \delta (\vartheta-\theta)-\theta_0]\}/\delta$ are c\`{a}dl\`{a}g for all $\theta \in \Theta_{\vnorm,M}$ and all $\delta \in [0,\Delta]$. Moreover, the following terms are all finite:
    $$\sup_{\theta \in \Theta_{\vnorm,M}} \| \ell(\theta) \|_{\vnorm}, \, \sup_{\theta \in \Theta_{\vnorm,M}} \| \ell_{0}'[\theta-\theta_0] \|_{\vnorm}, \, \sup_{\theta \in \Theta_{\vnorm,M}, \delta \in [0,\Delta]} \left\| \frac{r[\theta-\theta_0] - r[\theta -\theta_0 + \delta (\vartheta-\theta)]}{\delta} \right\|_{\vnorm}.$$
    In addition, $\| \ell_0'[\hat{\theta}_n - \theta_0] \|$ and $\sup_{\delta \in [0,\Delta]} \left\| \{r[\hat{\theta}_n-\theta_0] - r[\hat{\theta}_n -\theta_0 + \delta (\vartheta-\hat{\theta}_n)]\}/\delta \right\|$ converge to 0 in probability.
\end{conditionB}
\begin{conditionB}[Finite variance of influence function]
    \label{AHALfinitevar}
    $\xi^2 := \Var_{P_0} (\ell_{0}'[\dot{\Psi}](V))/\alpha_{0,\ell}^2 < \infty$.
\end{conditionB}

\subsection{Data-adaptive series} \label{section: series additional regularity conditions}

\begin{conditionC}[Local Lipschitz continuity of $\Pi_{n,\theta_0}(\id)$]
    \label{ALipschitzidentity}
    For sufficiently large $n$,
    $$\| \Pi_{n,\theta_0}(\id) \circ \theta - \Pi_{n,\theta_0}(\id) \circ \theta_0 \| \leq C \| \theta - \theta_0\|$$
    for all $\theta \in \Theta$ with $\| \theta - \theta_0 \| \leq n^{-1/4}$.
\end{conditionC}

\begin{conditionC}[Local Lipschitz continuity of $\dot{\psi}$ and $\Pi_{n,\theta_0}(\dot{\psi})$]
    \label{ALipschitzw}
    For sufficiently large $n$, for all $\theta \in \Theta$ with $\| \theta - \theta_0 \| \leq n^{-1/4}$,
    \begin{enumerate}[label=(\alph*)]
        \item \label{ALipschitzw first half} $\| \dot{\psi} \circ \theta - \dot{\psi} \circ \theta_0 \| \leq C \| \theta - \theta_0 \|$;
        \item $\| \Pi_{n,\theta_0}(\dot{\psi}) \circ \theta - \Pi_{n,\theta_0}(\dot{\psi}) \circ \theta_0 \| \leq C \| \theta - \theta_0\|$.
    \end{enumerate}
\end{conditionC}

\begin{conditionC}[Empirical process conditions]
    \label{Aempiricalprocess}
    There exists some constant $\Delta>0$ such that
    \begin{align*}
    & \sup_{\delta \in [0,\Delta]} \left| (P_n-P_0) \left\{ \frac{r[\theta_n^*-\theta_0] - r[\pi_n((1-\delta) \theta_n^* + \delta (\pm \dot{\Psi} + \theta_0))-\theta_0]}{\delta} \right\} \right|=o_p(n^{-1/2}), \\
    & (P_n-P_0) \ell_{0}'[(\pm \dot{\Psi} + \theta_0) - \pi_n(\pm \dot{\Psi} + \theta_0)]=o_p(n^{-1/2}), \\
    & (P_n-P_0) \ell_{0}'[\theta_n^*-\theta_0]=o_p(n^{-1/2}).
    \end{align*}
\end{conditionC}

\begin{conditionC}[Finite variance of influence function]
    \label{Afinitevar}
    $\xi^2 := \Var_{P_0} (\ell_{0}'[\dot{\Psi}](V))/\alpha_{0,\ell}^2 < \infty$.
\end{conditionC}

\subsection{Generalized data-adaptive series} \label{section: general series additional regularity conditions}

\begin{condition2}{ALipschitzidentity}[Local Lipschitz continuity of projected $\id$ for $\Theta_{n,\theta_0}$]
    \label{ALipschitzidentity2}
    For sufficiently large $n$, $\| \Pi_{n,\theta_0}(\id) \circ (\theta,\id_x) - \Pi_{n,\theta_0}(\id) \circ (\theta_0,\id_x) \| \leq C \| \theta - \theta_0\|$ for all $\| \theta - \theta_0 \| \leq n^{-1/4}$.
\end{condition2}

\begin{condition2}{ALipschitzw}[Local Lipschitz continuity of $\dot{\psi}$ and its projection for $\Theta_{n,\theta_0}$]
    \label{ALipschitzw2}
    For sufficiently large $n$, for all $\| \theta - \theta_0 \| \leq n^{-1/4}$,
    \begin{enumerate}[label=(\alph*)]
        \item \label{ALipschitzw2 first half} $\| \dot{\psi} \circ (\theta,\id_x) - \dot{\psi} \circ (\theta_0,\id_x) \| \leq C \| \theta - \theta_0 \|$;
        \item $\| \Pi_{n,\theta_0}(\dot{\psi}) \circ (\theta,\id_x) - \Pi_{n,\theta_0}(\dot{\psi}) \circ (\theta_0,\id_x) \| \leq C \| \theta - \theta_0\|$.
    \end{enumerate}
\end{condition2}

\subsection{Conditions for efficiency of the plug-in estimator} \label{section: regularity additional conditions}

Define a collection of submodels
$$\left\{ \{P_{H,\delta}: \delta \in B_H \subseteq \real\}: H \in \mathscr{H} \right\}$$
for which: (i) $\mathscr{H}$ is a subset of $L_0^2(P_0)$ and the $L_0^2(P_0)$-closure of its linear span is $L_0^2(P_0)$; and (ii) each $\{P_{H,\delta}: \delta \in B_H \subseteq \real\}$ is a regular univariate parametric submodel that passes through $P_0$ and has score $H$ for $\delta$ at $\delta=0$. For each $H \in \mathscr{H}$ and $\delta \in B_H$, we define $\theta_{H,\delta} \in \argmin_{\theta \in \Theta} P_{H,\delta} \ell(\theta)$. In this appendix, for all small $o$ and big $O$ notations, we let $\delta \rightarrow 0$ with $H$ fixed.

\begin{conditionE}[Sufficiently close risk minimizer]
    \label{Acloseminimizer}
    For any given $H \in \mathscr{H}$, $\|\theta_{H,\delta} - \theta_0\| = o(\delta^{1/2})$.
\end{conditionE}

\begin{conditionE}[Quadratic behavior of loss function remainder near 0]
    \label{Aempiricalprocesslike}
    For any given $H \in \mathscr{H}$ and $\vartheta$, there exists positive $\delta'=o(\delta)$ such that $(P_{H,\delta} - P_0) \{ r[(1-\delta')(\theta_{H,\delta} - \theta_0) + \delta' \dot{\Psi}] - r[\theta_{H,\delta} - \theta_0] \}/\delta' = o(\delta)$.
\end{conditionE}

\section{Discussion of technical conditions for data-adaptive series and its generalization} \label{section: condition discussion}

\subsection{Theorem~\ref{Tefficiency}}

\bfitem{Condition~\ref{Aestimation}} usually imposes an upper bound on the growth rate of $K$. To see this, we show that Condition~\ref{Aestimation} is equivalent to a term being $o_p(n^{-1/4})$, and an upper bound of this term is controlled by $K$. Let $\theta_n^\dagger \in \argmin_{\theta \in \Theta_n} P_0 \ell(\theta)$ be the true-risk minimizer in $\Theta_n$. Under Conditions~\ref{Aquadraticloss}, \ref{Ainit}, \ref{Aapproxidentity} and \ref{ALipschitzidentity}, by Lemma~\ref{Lestimation}, it follows that Condition~\ref{Aestimation} is equivalent to requiring that $\| \theta_n^* - \theta_n^\dagger \|=o_p(n^{-1/4})$. Note that $\theta_n^*$ minimizes the empirical risk in $\Theta_n$, and M-estimation theory \citep{vandervaart1996} can show that $\| \theta_n^* - \theta_n^\dagger \|$ can be upper bounded by an empirical process term, whose upper bound is related to the complexity of $\Theta_n$, namely how fast $K$ grows with sample size. To ensure this bound is $o_p(n^{-1/4})$, $K$ must not grow too quickly.

\vspace{0.5\baselineskip}

\bfitem{Condition~\ref{Aapproxidentity}} assumes that the identity function can be well approximated by the series $\phi_k$ with the specified number of terms $K$ in the $L^2(P_{\theta_0})$ sense. If $\text{Span}\{\phi_1,\ldots,\phi_K\}$ does not contain $\id$ for any $K$, then sufficiently many terms must be included to satisfy this condition; that is, this condition imposes a lower bound on the rate at which $K$ should grow with $n$. Even if $\text{Span}\{\phi_1,\ldots,\phi_K\}$ does contain $\id$ for some finite $K$, this condition still requires that $K$ is not too small.

\vspace{0.5\baselineskip}

\bfitem{Condition~\ref{Aapproxw}} is implied by the following condition in view of Lemma~\ref{Lrate}:
\begin{conditionsufficient}{Aapproxw}
    \label{Aapproxwsufficient}
    $\| [\dot{\psi} - \Pi_{n,\theta_0}(\dot{\psi})] \circ \theta_0 \|=o(n^{-1/4})$.
\end{conditionsufficient}
This condition is similar to Condition~\ref{Aapproxidentity}. However, in general, we do not expect $\dot{\psi}$ to be contained in $\text{Span}\{\phi_1,\ldots,\phi_K\}$ for any $K$, and hence this condition generally imposes a lower bound on the rate of $K$. Note that Condition~\ref{Aapproxwsufficient} is stronger than Condition~\ref{Aapproxw}, and there are interesting examples where \ref{Aapproxw} holds but \ref{Aapproxwsufficient} fails to hold. Indeed, if $\theta_n^*$ converges to $\theta_0$ at a rate much faster than $n^{-1/4}$, then \ref{Aapproxw} can be satisfied even if $\| [\dot{\psi} - \Pi_{n,\theta_0}(\dot{\psi})] \circ \theta_0 \|$ decays to zero in probability relatively slowly --- that is, the convergence rate of $\theta_n^*$ can compensate for the approximation error of $\dot{\psi}$. This is one way in which we can benefit from using flexible ML algorithms to estimate $\theta_0$: if $\theta_n^0$ converges to $\theta_0$ at a fast rate, then we can expect $\theta_n^*$ to also have a fast convergence rate.

\vspace{0.5\baselineskip}

\bfitem{Conditions~\ref{Aestimation}, \ref{Aapproxidentity} and \ref{Aapproxw}} are not stringent provided sufficient smoothness on derivatives of $\dot{\psi}$ and a reasonable series. For example, as noted in \cite{Chen2007}, when $\dot{\psi}$ has a bounded $p$-th order derivative and the polynomial, trigonometric series or spline with degree at least $p+1$ is used, then if $K^2/n \rightarrow 0$ ($K^3/n \rightarrow 0$ for polynomial series), the term in Condition~\ref{Aestimation} is $O_p(\sqrt{K/n})$; the terms in Condition~\ref{Aapproxidentity} and the sufficient Condition~\ref{Aapproxwsufficient} are $O(K^{-p/q})$. Therefore, we can select $K$ to grow at a rate faster than $n^{q/(4p)}$ and slower than $n^{1/2}$ ($n^{1/3}$ for polynomial series). If $p$ is large, then this allows for a wide range of rates for $K$. Typically $\dot{\Psi}$ (and hence $\dot{\psi}$) is only related to the summary of interest $\Psi$ but not the true function $\theta_0$. For example, for the summary $\Psi(\theta)=P_0 (f \circ \theta)$ at the beginning of Section~\ref{section: simple case method}, $\dot{\psi}=f'$ is variation independent of $\theta_0$. It is often the case that $\Psi$ is smooth and so is $\dot{\psi}$, so $p$ is often sufficiently large for this window to be wide.

\vspace{0.5\baselineskip}

\bfitem{Condition~\ref{ALipschitzidentity}} is usually easy to satisfy. Since $\Pi_{n,\theta_0}(\id)$ is a linear combination of $\{\phi_k: k \in \{1,\ldots,K\} \}$ and is an approximation of a highly smooth function $\id$, if the series $\phi_k$ is smooth, then we can expect that $\Pi_{n,\theta_0}(\id)$ will be Lipschitz uniformly over $n$, that is, that Condition~\ref{ALipschitzidentity} holds. For example, using polynomial series, cubic splines or trigonometric series imply that this condition holds.

\vspace{0.5\baselineskip}

\bfitem{Condition~\ref{ALipschitzw}} imposes Lipschitz continuity conditions on $\dot{\psi}$ and $\Pi_{n,\theta_0}(\dot{\psi})$ uniformly over $n$. The Lipschitz continuity of $\dot{\psi}$ has been discussed above. As for $\Pi_{n,\theta_0}(\dot{\psi})$, similarly to Condition~\ref{ALipschitzidentity}, as long as the series $\phi_k$ being used is smooth, $\Pi_{n,\theta_0}(\dot{\psi})$ would be Lipschitz continuous uniformly over $n$.

\subsection{Theorem~\ref{Tefficiency2}}

The conditions are similar to those in Theorem~\ref{Tefficiency}. However, Condition~\ref{Aapproxw2} can be more stringent than Condition~\ref{Aapproxw}. For generalized data-adaptive series, the dimension of the argument of the series is larger. Hence, as noted in \cite{Chen2007}, \ref{Aapproxw2} may require more smoothness of $\dot{\psi}$ in order that $\dot{\psi}$ can be well approximated by $\Pi_{n,\theta_0}(\dot{\psi})$. However, in general, we do not expect the smoothness of $\dot{\psi}$ to depend on $\Psi$ alone but no components of $P_0$, so the amount of smoothness of $\dot{\psi}$ may be more limited in practice.

It is also worth noting that, similarly to Theorem~\ref{Tefficiency}, a sufficient condition for Condition~\ref{Aapproxw2} is the following:
\begin{conditionsufficient2}{Aapproxw}
    \label{Aapproxw2sufficient}
    $\| [\dot{\psi} - \Pi_{n,\theta_0}(\dot{\psi})] \circ (\theta_0,\id_x) \| = o(n^{-1/4})$.
\end{conditionsufficient2}

\section{Lemmas and technical proofs} \label{appendix: proof}

\subsection{Highly Adaptive Lasso (HAL)} \label{appendix: HAL}

\begin{proof}[Proof of Theorem~\ref{THALefficiency}]
    Under Conditions~\ref{Aquadraticloss} and \ref{Acadlag}--\ref{AHALempiricalprocess}, Lemma~1 and its corollary in \cite{VanderLaan2017} show that $\| \hat{\theta}_n - \theta_0 \|=o_p(n^{-1/4})$.
    
    We show that the small perturbations of $\hat{\theta}_n$ in certain directions are contained in $\Theta_{\vnorm,M}$. Let $\vartheta_\delta=\hat{\theta}_n+\delta (\dot{\Psi}+\theta_0-\hat{\theta}_n)$ be a path indexed by $\delta$ $(0 \leq \delta < 1)$ that is a perturbation of $\hat{\theta}_n$. Note that for all $\delta$, $\vartheta_\delta$ is c\`{a}dl\`{a}g by Condition~\ref{Acadlag} and we have that
    $$\| \vartheta_\delta \|_{\vnorm} = \| (1-\delta) \hat{\theta}_n + \delta (\dot{\Psi} + \theta_0) \|_{\vnorm} \leq (1-\delta) \| \hat{\theta}_n \|_{\vnorm} + \delta (\| \dot{\Psi} \|_{\vnorm} + \| \theta_0 \|_{\vnorm}) \leq (1-\delta) M + \delta M =M$$
    by Condition~\ref{AM}. Hence $\vartheta_\delta \in \Theta_{\vnorm,M}$. The same result holds for the path $\tilde{\vartheta}_\delta := \hat{\theta}_n+\delta (-\dot{\Psi}+\theta_0-\hat{\theta}_n)$.
    
    Combining this observation with the $P_0$-Donkser property of $\Theta_{\vnorm,M'}$ for any fixed $M'>0$ \citep{gill1993} and Conditions~\ref{Adloss}--\ref{Aquadraticloss}, \ref{AHALfinitevar}, we have that all of the conditions of Theorem~1 in \cite{Shen1997} are satisfied with all sieves being $\Theta_{\vnorm,M}$. The desired asymptotic linearity result follows. The efficiency result is shown in Appendix~\ref{appendix: regular proof}.
\end{proof}

\begin{proof}[Proof of Lemma~\ref{Lvarnorm}]
    Recall that $\mathcal{X} \subseteq \real^d$. Similar to $x^{(l)}$, let $x^{(u)}=\inf \{x: P_0(X \leq x) = 1\}$ where $\inf$ and $\leq$ are entrywise. To avoid clumsy notations, in this proof we drop the subscript in $\theta_0$ and use $\theta$ instead. This should not introduce confusion because other functions (e.g., an estimator of $\theta_0$) are not involved in the statement or proof. Using the results reviewed in Section~\ref{section: hal review},
    \begin{align*}
    \| \dot{\Psi} \|_{\vnorm} &= | \dot{\Psi}(x^{(\ell)}) | + \sum_{s \subseteq \{1,2,\ldots,d\}, s \neq \emptyset} \int_{x^{(\ell)}_s}^{x^{(u)}_s} | \dot{\Psi}_s (du) | \\
    &= | \dot{\Psi}(x^{(\ell)}) | + \sum_{s \subseteq \{1,2,\ldots,d\}, s \neq \emptyset} \int_{x^{(\ell)}_s}^{x^{(u)}_s} | \dot{\psi}'(z)| \Big|_{z=\theta_s(u)} | \theta_s (du) |.
    \end{align*}
    Since
    \begin{align*}
    |\theta(x)| &= \left| \theta(x^{(\ell)}) + \sum_{s \subseteq \{1,2,\ldots,d\}, s \neq \emptyset} \int_{x^{(\ell)}_s}^{x_s} \theta_s (du) \right| \\
    &\leq | \theta(x^{(\ell)}) | + \sum_{s \subseteq \{1,2,\ldots,d\}, s \neq \emptyset} \int_{x^{(\ell)}_s}^{x_s} | \theta_s (du) | \\
    &\leq | \theta(x^{(\ell)}) | + \sum_{s \subseteq \{1,2,\ldots,d\}, s \neq \emptyset} \int_{x^{(\ell)}_s}^{x^{(u)}_s} | \theta_s (du) | = \| \theta \|_{\vnorm},
    \end{align*}
    we have $| \dot{\psi}'(z)| \Big|_{z=\theta_s(u)} \leq \sup_{z': |z'| \leq \| \theta_0 \|_{\vnorm}} |\dot{\psi}'(z')| =B$ for all $x^{(\ell)} \leq u \leq x^{(u)}$, so
    $$\| \dot{\Psi} \|_{\vnorm} \leq | \dot{\Psi}(x^{(\ell)}) | + \sum_{s \subseteq \{1,2,\ldots,d\}, s \neq \emptyset} \int_{x^{(\ell)}_s}^{x^{(u)}_s} B | \theta_s (du) | \leq | \dot{\Psi}(x^{(\ell)}) | + B \| \theta_0 \|_{\vnorm}.$$
\end{proof}

\begin{lemma}[CV-selected bound not much smaller than the bound of the true function's variation norm]
    \label{Lcvbound}
    Suppose that Condition~\ref{Acadlag} holds, $\theta_0$ is c\`{a}dl\`{a}g, $\| \theta_0 \|_{\vnorm}<\infty$ and for any $M$, $\sup_{\theta \in \Theta_{\vnorm,M}} \| \ell(\theta) \| < \infty$. Let $M_n$ be a (possibly random) sequence such that $P_0 \{ \ell(\hat{\theta}_{n,M_n}) - \ell(\theta_0) \}=o_p(1)$. Then for any $\epsilon>0$, with probability tending to one, $M_n \geq \| \theta_0 \|_{\vnorm} - \epsilon$. Therefore, for any fixed $\epsilon>0$, with probability tending to one, $M_n + \epsilon \geq (\| \theta_0 \|_{\vnorm} - \epsilon) + \epsilon = \| \theta_0 \|_{\vnorm}$.
\end{lemma}
\begin{proof}[Proof of Lemma~\ref{Lcvbound}]
    We prove by contradiction. Suppose the claim is not true, i.e. there exists $\epsilon, \delta > 0$ such that $P(M_n < \| \theta_0 \|_{\vnorm} - \epsilon) \geq \delta$ for all $n \in \mathcal{N}$, where $\mathcal{N}$ is an infinite set. Let $\theta_{0,M} \in \argmin_{\theta \in \Theta_{\vnorm,M}} P_0 \ell(\theta)$. Then for all $n \in \mathcal{N}$, with probability at least $\delta$,
    \begin{align*}
    P_0 \{ \ell(\hat{\theta}_{n,M_n}) - \ell(\theta_0) \} &= P_0 \{ \ell(\hat{\theta}_{n,M_n}) - \ell(\theta_{0,M_n}) \} + P_0 \{ \ell(\theta_{0,M_n}) - \ell(\theta_0) \} \\
    &\geq P_0 \{ \ell(\theta_{0,M_n}) - \ell(\theta_0) \} \\
    &\geq P_0 \{ \ell(\theta_{0,\| \theta_0 \|_{\vnorm} - \epsilon}) - \ell(\theta_0) \},
    \end{align*}
    which is a positive constant since the function class $\Theta_{\|\theta_0\|_{\vnorm} - \epsilon}$ does not contain $\theta_0$ and this term is non-negligible bias. This contradicts the assumption that $P_0 \{ \ell(\hat{\theta}_{n,M_n}) - \ell(\theta_0) \}=o_p(1)$ and hence the desired follows.
\end{proof}

Therefore, if $\| \dot{\Psi} \|_{\vnorm} \leq F(\| \theta_0 \|_{\vnorm})$ for a known increasing function $F$, then with probability tending to one, $M_n + \epsilon + F(M_n+\epsilon)$ is a valid bound on $\| \hat{\theta}_n \|_\vnorm$ that can be used to obtain an efficient plug-in estimator. Moreover, if the bound is loose, i.e. $\| \dot{\Psi} \|_{\vnorm} < F(\| \theta_0 \|_{\vnorm})$, and $F$ is continuous, then there exists some $\epsilon>0$ such that $\| \dot{\Psi} \|_{\vnorm} \leq F(\| \theta_0 \|_{\vnorm} - \epsilon) - \epsilon$ and hence $\|\theta_0\|_{\vnorm} + \| \dot{\Psi} \|_{\vnorm} \leq M_n + F(M_n)$ with probability tending to one.

Note that this lemma only concerns learning a function-valued feature but not estimating $\Psi(\theta_0)$. There are examples where $\dot{\Psi}$ depends on components of $P_0$, say $\eta_0$, other than $\theta_0$. However, if $\eta_0$ can be learned via HAL, then Lemma~\ref{Lcvbound} can be applied. Therefore, if it is known that $\| \dot{\Psi} \|_{\vnorm} \leq F(\| \theta_0 \|_{\vnorm}, \| \eta_0 \|_{\vnorm})$ for a known increasing function $F$, then we can use a bound on $\| \hat{\theta}_n \|_\vnorm$ obtained in a similar fashion as above from the sequence $M_n$ to construct an efficient plug-in estimator $\Psi(\hat{\theta}_n)$.

Now consider obtaining $M_n$ by $k$-fold CV from a set of candidate bounds. Then, under Conditions~\ref{Acadlag}--\ref{AHALempiricalprocess}, by (i) Lemma 1 and its corollary of \cite{VanderLaan2017}, and (ii) the oracle inequality for $k$-fold CV in \cite{Vanderlaan2003cv}, $P_0 \{ \ell(\hat{\theta}_{n,M_n}) - \ell(\theta_0) \}=o_p(n^{-1/4})$ if (i) one candidate bound is no smaller than $\| \theta_0 \|_\vnorm$, and (ii) the number of candidate bounds is fixed. Therefore, the above results apply to this case.

\subsection{Data-adaptive series estimation} \label{appendix: sieve}

We first present and prove two lemmas that lead to Theorems~\ref{Tefficiency} and \ref{Tefficiency2}.

\begin{lemma}[Convergence rate of the sieve estimator]
    \label{Lrate}
    Under Conditions~\ref{Ainit}, \ref{Aapproxidentity} and \ref{ALipschitzidentity}, $\| \pi_n(\theta_0) - \theta_0 \|=o_p(n^{-1/4})$. Under an additional condition \ref{Aestimation}, $\| \theta_n^* -\theta_0 \|=o_p(n^{-1/4})$.
\end{lemma}

\begin{proof}[Proof of Lemma~\ref{Lrate}]
    By triangle inequality, $\| \pi_n(\theta_0)-\theta_0 \| \leq \| \theta_0 - \theta_n^0 \| + \| \theta_n^0 - \pi_n(\theta_n^0) \| + \| \pi_n(\theta_n^0) - \pi_n(\theta_0) \|$. We bound these three terms separately.
    
    \bfitem{Term 1}: By Condition~\ref{Ainit}, $\| \theta_0 - \theta_n^0 \|=o_p(n^{-1/4})$.
    
    \bfitem{Term 2}: By the definition of projection operator,
    $$\| \theta_n^0 - \pi_n(\theta_n^0) \| = \| \theta_n^0 - \Pi_{n,\theta_n^0} (\id) \circ \theta_n^0 \| \leq \| \theta_n^0 - \Pi_{n,\theta_0} (\id) \circ \theta_n^0 \|.$$
    We bound the right-hand side by showing this term is close to $\| \theta_0 - \Pi_{n,\theta_0}(\id) \circ \theta_0 \|$ up to an $o_p(n^{-1/4})$ term. By the reverse triangle inequality and the triangle inequality,
    \begin{align*}
    & \left| \| \theta_n^0 - \Pi_{n,\theta_0} (\id) \circ \theta_n^0 \| - \| \theta_0 - \Pi_{n,\theta_0} (\id) \circ \theta_0 \| \right| \\
    &\quad \leq \| [\theta_n^0 - \Pi_{n,\theta_0} (\id) \circ \theta_n^0] - [\theta_0 - \Pi_{n,\theta_0} (\id) \circ \theta_0] \| \\
    &\quad = \| [\theta_n^0 - \theta_0] - [\Pi_{n,\theta_0} (\id) \circ \theta_n^0 - \Pi_{n,\theta_0}(\id) \circ \theta_0] \| \\
    &\quad \leq \| \theta_n^0 - \theta_0 \| + \| \Pi_{n,\theta_0} (\id) \circ \theta_n^0 - \Pi_{n,\theta_0}(\id) \circ \theta_0 \| \\
    &\quad \leq \| \theta_n^0 - \theta_0 \| + C \| \theta_n^0 - \theta_0 \|, & \text{(Condition~\ref{ALipschitzidentity})}
    \end{align*}
    which is $o_p(n^{-1/4})$ by Condition~\ref{Ainit}. Therefore, by Condition~\ref{Aapproxidentity},
    $$\| \theta_n^0 - \pi_n(\theta_n^0) \| \leq \| \theta_n^0 - \Pi_{n,\theta_0} (\id) \circ \theta_n^0 \| \leq \| \theta_0 - \Pi_{n,\theta_0} (\id) \circ \theta_0 \| + o_p(n^{-1/4}) = o_p(n^{-1/4}).$$
    
    \bfitem{Term 3}: By the definition of projection and Condition~\ref{Ainit}, $\| \pi_n(\theta_n^0) - \pi_n(\theta_0) \| \leq \| \theta_n^0 - \theta_0 \| = o_p(n^{-1/4})$.
    
    \bfitem{Conclusion from the three bounds}: $\| \pi_n(\theta_0) - \theta_0 \|=o_p(n^{-1/4})$.
    
    If, in addition, Condition~\ref{Aestimation} also holds, then $\| \theta_n^* -\theta_0 \| \leq \| \pi_n(\theta_0) - \theta_0 \| + \| \theta_n^* - \pi_n(\theta_0) \| = o_p(n^{-1/4})$.
\end{proof}

The same result holds for the generalized data-adaptive series under Conditions~\ref{Ainit}, \ref{ALipschitzidentity2}, \ref{Aapproxidentity2} and \ref{Aestimation} (if relevant). The proof is almost identical and is therefore omitted.

\begin{lemma}[Approximation error to $\dot{\psi}$]
    \label{Lapproxw}
    Under Condition~\ref{ALipschitzw}, $\| \dot{\psi} \circ \theta_0 - \pi_n (\dot{\psi} \circ \theta_0) \| \leq C \| \theta_n^0-\theta_0 \| + \| \dot{\psi} \circ \theta_0 - \Pi_{n,\theta_0}(\dot{\psi}) \circ \theta_0 \|$. Therefore, under Conditions~\ref{Ainit}--\ref{Aapproxw}, $\| \dot{\psi} \circ \theta_0 - \pi_n (\dot{\psi} \circ \theta_0) \| \cdot \| \theta_n^*-\theta_0 \|=o_p(n^{-1/2})$.
\end{lemma}

\begin{proof}[Proof of Lemma~\ref{Lapproxw}]
    By the definition of the projection operator and triangle inequality,
    $$\| \dot{\psi} \circ \theta_0 - \pi_n (\dot{\psi} \circ \theta_0) \| \leq \| \dot{\psi} \circ \theta_0 - \pi_n (\dot{\psi} \circ \theta_n^0) \| \leq \| \dot{\psi} \circ \theta_0 - \dot{\psi} \circ \theta_n^0 \| + \| \dot{\psi} \circ \theta_n^0 - \pi_n (\dot{\psi} \circ \theta_n^0) \|.$$
    We bound the two terms on the right-hand side separately.
    
    \bfitem{Term 1}: By Condition~\ref{ALipschitzw}, $\| \dot{\psi} \circ \theta_0 - \dot{\psi} \circ \theta_n^0 \| \leq C \| \theta_0 - \theta_n^0 \|$.
    
    \bfitem{Term 2}: This term can be bounded similarly as in Lemma~\ref{Lrate}. By the reverse triangle inequality and the triangle inequality,
    \begin{align*}
    & \left| \| \dot{\psi} \circ \theta_n^0 - \Pi_{n,\theta_0} (\dot{\psi}) \circ \theta_n^0 \| - \| \dot{\psi} \circ \theta_0 - \Pi_{n,\theta_0} (\id) \circ \theta_0 \| \right| \\
    &\quad \leq \| [\dot{\psi} \circ \theta_n^0 - \Pi_{n,\theta_0} (\dot{\psi}) \circ \theta_n^0] - [\dot{\psi} \circ \theta_0 - \Pi_{n,\theta_0}(\dot{\psi}) \circ \theta_0] \| \\
    &\quad = \| [\dot{\psi} \circ \theta_n^0 - \dot{\psi} \circ \theta_0] - [\Pi_{n,\theta_0} (\dot{\psi}) \circ \theta_n^0 - \Pi_{n,\theta_0}(\dot{\psi}) \circ \theta_0] \| \\
    &\quad \leq \| \dot{\psi} \circ \theta_n^0 - \dot{\psi} \circ \theta_0 \| + \| \Pi_{n,\theta_0} (\dot{\psi}) \circ \theta_n^0 - \Pi_{n,\theta_0}(\dot{\psi}) \circ \theta_0 \| \\
    &\quad \leq C \| \theta_n^0 - \theta_0 \| + C \| \theta_n^0 - \theta_0 \| & \text{(Condition~\ref{ALipschitzw})} \\
    &\quad = C \| \theta_n^0 - \theta_0 \|.
    \end{align*}
    Therefore, by the definition of the projection operator and Condition~\ref{ALipschitzw},
    \begin{align*}
    \| \dot{\psi} \circ \theta_n^0 - \pi_n (\dot{\psi} \circ \theta_n^0) \| &\leq \| \dot{\psi} \circ \theta_n^0 - \Pi_{n,\theta_0} (\dot{\psi}) \circ \theta_n^0 \| \\
    &\leq \| \dot{\psi} \circ \theta_0 - \Pi_{n,\theta_0} (\dot{\psi}) \circ \theta_0 \| + C \| \theta_n^0 - \theta_0 \|.
    \end{align*}
    
    \bfitem{Conclusion from the two bounds}: $\| \dot{\psi} \circ \theta_0 - \pi_n (\dot{\psi} \circ \theta_0) \| \leq C \| \theta_n^0-\theta_0 \| + \| \dot{\psi} \circ \theta_0 - \Pi_{n,\theta_0}(\dot{\psi}) \circ \theta_0 \|$.
    
    Under Conditions~\ref{Ainit}--\ref{Aapproxw}, using Lemma~\ref{Lrate}, it follows that $\| \dot{\psi} \circ \theta_0 - \pi_n (\dot{\psi} \circ \theta_0) \| \cdot \| \theta_n^*-\theta_0 \|=o_p(n^{-1/2})$.
\end{proof}

Note that $\pi_n$ is a linear operator. Lemma~\ref{Lrate}~and~\ref{Lapproxw} along with other conditions essentially satisfy the assumptions in Corollary~2 in \cite{Shen1997}. We can prove the asymptotic linearity result of Theorem~\ref{Tefficiency} similarly to this result as follows.

\begin{proof}[Proof of Theorem~\ref{Tefficiency}]
    We note that
    \begin{align*}
    P_n \ell(\theta_n^*) &= P_n \ell(\theta_0) + P_0 [\ell(\theta_n^*) - \ell(\theta_0)] + (P_n-P_0) [\ell(\theta_n^*) - \ell(\theta_0)] \\
    &= P_n \ell(\theta_0) + P_0 [\ell(\theta_n^*) - \ell(\theta_0)] + (P_n-P_0) \ell_0'[\theta_n^* - \theta_0] \\
    &\quad+ (P_n-P_0) r[\theta_n^* - \theta_0].
    \end{align*}
    Let $\epsilon_n$ be an arbitrary sequence of positive real numbers that is $o(n^{-1/2})$. We may replace $\theta_n^*$ with $\pi_n((1-\epsilon_n) \theta_n^* + \epsilon_n (\theta_0 \pm \dot{\Psi}))$ in the above equation. We first consider $\pi_n((1-\epsilon_n) \theta_n^* + \epsilon_n (\theta_0 + \dot{\Psi}))$:
    \begin{align}
    \begin{split}
    \label{eq: asymptotic linearity key2}
    & P_n \ell \left( \pi_n((1-\epsilon_n) \theta_n^* + \epsilon_n (\theta_0 + \dot{\Psi})) \right) \\
    &= P_n \ell(\theta_0) + P_0 [\ell( \pi_n((1-\epsilon_n) \theta_n^* + \epsilon_n (\theta_0 + \dot{\Psi})) ) - \ell(\theta_0)] \\
    &\quad+ (P_n-P_0) \ell_0'[ \pi_n((1-\epsilon_n) \theta_n^* + \epsilon_n (\theta_0 + \dot{\Psi})) - \theta_0] \\
    &\quad+ (P_n-P_0) r[ \pi_n((1-\epsilon_n) \theta_n^* + \epsilon_n (\theta_0 + \dot{\Psi})) - \theta_0].
    \end{split}
    \end{align}
    Take the difference between the above two equations. By the linearity of $\ell_0'$ and $\pi_n$, we have that
    \begin{align*}
    & P_n \ell \left( \pi_n((1-\epsilon_n) \theta_n^* + \epsilon_n (\theta_0 + \dot{\Psi})) \right) - P_n \ell(\theta_n^*) \\
    &= P_0 [\ell(\pi_n((1-\epsilon_n) \theta_n^* + \epsilon_n (\theta_0 + \dot{\Psi}))) - \ell(\theta_0)] - P_0 [\ell(\theta_n^*) - \ell(\theta_0)] \\
    &\quad+ (P_n-P_0)\ell_0'[\pi_n((1-\epsilon_n) \theta_n^* + \epsilon_n (\theta_0 + \dot{\Psi})) - \theta_n^*] \\
    &\quad+ (P_n-P_0) \{ r[\pi_n((1-\epsilon_n) \theta_n^* + \epsilon_n (\theta_0 + \dot{\Psi})) - \theta_0] - r[\theta_n^* - \theta_0] \}.
    \end{align*}
    We next analyze the three lines on the right-hand side of the above equation separately.
    
    \bfitem{Line~1}: Under Condition~\ref{Aquadraticloss},
    \begin{align*}
    & P_0 [\ell(\pi_n((1-\epsilon_n) \theta_n^* + \epsilon_n (\theta_0 + \dot{\Psi}))) - \ell(\theta_0)] - P_0 [\ell(\theta_n^*) - \ell(\theta_0)] \\
    &= \frac{\alpha_{0,\ell}}{2} \| \pi_n((1-\epsilon_n) \theta_n^* + \epsilon_n (\theta_0 + \dot{\Psi})) - \theta_0 \|^2 - \frac{\alpha_{0,\ell}}{2} \| \theta_n^* - \theta_0 \|^2 \\
    &\quad+ o_p \left( \| \pi_n((1-\epsilon_n) \theta_n^* + \epsilon_n (\theta_0 + \dot{\Psi})) - \theta_0 \|^2 + \| \theta_n^* - \theta_0 \|^2 \right) \\
    \intertext{We subtract and add $(1-\epsilon_n) \theta_n^* + \epsilon_n (\theta_0 + \dot{\Psi})$ in the first term. By the fact that $\pi_n$ is linear and $\pi_n(\theta_n^*)=\theta_n^*$, the display continues as}
    &= \frac{\alpha_{0,\ell}}{2} \| \{\pi_n((1-\epsilon_n) \theta_n^* + \epsilon_n (\theta_0 + \dot{\Psi})) - ((1-\epsilon_n) \theta_n^* + \epsilon_n (\theta_0 + \dot{\Psi}))\} \\
    &\quad+ \{(1-\epsilon_n) \theta_n^* + \epsilon_n (\theta_0 + \dot{\Psi}) - \theta_0\} \|^2 \\
    &\quad- \frac{\alpha_{0,\ell}}{2} \| \theta_n^* - \theta_0 \|^2 + o_p \left( \| \pi_n((1-\epsilon_n) \theta_n^* + \epsilon_n (\theta_0 + \dot{\Psi})) - \theta_0 \|^2 + \| \theta_n^* - \theta_0 \|^2 \right) \\
    &= \frac{\alpha_{0,\ell}}{2} \| \epsilon_n \{\pi_n(\theta_0 + \dot{\Psi}) - (\theta_0 + \dot{\Psi})\} + (\theta_n^*-\theta_0) + \epsilon_n (\dot{\Psi}+\theta_0-\theta_n^*) \|^2 - \frac{\alpha_{0,\ell}}{2} \| \theta_n^* - \theta_0 \|^2 \\
    &\quad+ o_p \left( \| \pi_n((1-\epsilon_n) \theta_n^* + \epsilon_n (\theta_0 + \dot{\Psi})) - \theta_0 \|^2 + \| \theta_n^* - \theta_0 \|^2 \right) \\
    &= \epsilon_n \alpha_{0,\ell} \langle \theta_n^*-\theta_0, \dot{\Psi} \rangle + \epsilon_n^2 \frac{\alpha_{0,\ell}}{2} \| \pi_n(\theta_0+\dot{\Psi}) - \theta_n^* \|^2 \\
    &\quad+ \epsilon_n \alpha_{0,\ell} \langle \pi_n(\theta_0) - \theta_0, \theta_n^* - \theta_0 \rangle + \epsilon_n \alpha_{0,\ell} \langle \pi_n(\dot{\Psi}) - \dot{\Psi}, \theta_n^* - \theta_0 \rangle - \epsilon_n \alpha_{0,\ell} \| \theta_n^* - \theta_0 \|^2 \\
    &\quad+ o_p \left( \| \pi_n((1-\epsilon_n) \theta_n^* + \epsilon_n (\theta_0 + \dot{\Psi})) - \theta_0 \|^2 + \| \theta_n^* - \theta_0 \|^2 \right) \\
    \intertext{By Cauchy-Schwards inequality, the display continues as}
    &\leq \epsilon_n \alpha_{0,\ell} \langle \theta_n^*-\theta_0, \dot{\Psi} \rangle + \epsilon_n^2 \frac{\alpha_{0,\ell}}{2} \| \pi_n(\theta_0+\dot{\Psi}) - \theta_n^* \|^2 \\
    &\quad+ \epsilon_n \alpha_{0,\ell} \| \pi_n(\theta_0) - \theta_0 \| \| \theta_n^* - \theta_0 \| + \epsilon_n \alpha_{0,\ell} \| \pi_n(\dot{\Psi}) - \dot{\Psi} \| \| \theta_n^* - \theta_0 \| - \epsilon_n \alpha_{0,\ell} \| \theta_n^* - \theta_0 \|^2 \\
    &\quad+ o_p \left( \| \pi_n((1-\epsilon_n) \theta_n^* + \epsilon_n (\theta_0 + \dot{\Psi})) - \theta_0 \|^2 + \| \theta_n^* - \theta_0 \|^2 \right) \\
    \intertext{By Lemmas~\ref{Lrate}--\ref{Lapproxw} and the assumption that $\epsilon_n=o(n^{-1/2})$, the display continues as}
    &= \epsilon_n \alpha_{0,\ell} \langle \theta_n^*-\theta_0, \dot{\Psi} \rangle + \epsilon_n o_p(n^{-1/2}) + o_p \left( \| \pi_n((1-\epsilon_n) \theta_n^* + \epsilon_n (\theta_0 + \dot{\Psi})) - \theta_0 \|^2 + \| \theta_n^* - \theta_0 \|^2 \right).
    \end{align*}
    \bfitem{Line~2}: We subtract and add $(1-\epsilon_n) \theta_n^* + \epsilon_n (\theta_0 + \dot{\Psi})$. By linearity of $\ell_0'$, Condition~\ref{Aempiricalprocess}, and the fact that $\pi_n(\theta_n^*)=\theta_n^*$, we have that
    \begin{align*}
    & (P_n-P_0)\ell_0'[\pi_n((1-\epsilon_n) \theta_n^* + \epsilon_n (\theta_0 + \dot{\Psi})) - \theta_n^*] \\
    &= (P_n-P_0)\ell_0'[(1-\epsilon_n) \theta_n^* + \epsilon_n (\theta_0 + \dot{\Psi}) - \theta_n^*] \\
    &\quad+ \epsilon_n (P_n-P_0) \ell_0'[\pi_n(\theta_0+\dot{\Psi}) - (\theta_0+\dot{\Psi})] \\
    &= \epsilon_n (P_n-P_0) \ell_0'[\dot{\Psi}] - \epsilon_n (P_n-P_0) \ell_0'[\theta_n^* - \theta_0] + \epsilon_n o_p(n^{-1/2}) \\
    &= \epsilon_n (P_n-P_0) \ell_0'[\dot{\Psi}] + \epsilon_n o_p(n^{-1/2}).
    \end{align*}
    \bfitem{Line~3}: By Condition~\ref{Aempiricalprocess}, this term is $\epsilon_n o_p(n^{-1/2})$.
    
    \bfitem{Conclusion of the three lines}: It holds that
    \begin{align*}
    & P_n \ell \left( \pi_n((1-\epsilon_n) \theta_n^* + \epsilon_n (\theta_0 + \dot{\Psi})) \right) - P_n \ell(\theta_n^*) \\
    &\leq \epsilon_n \alpha_{0,\ell} \langle \theta_n^*-\theta_0, \dot{\Psi} \rangle + \epsilon_n (P_n-P_0) \ell_0'[\dot{\Psi}] \\
    &\quad+ \epsilon_n o_p(n^{-1/2}) + o_p \left( \| \pi_n((1-\epsilon_n) \theta_n^* + \epsilon_n (\theta_0 + \dot{\Psi})) - \theta_0 \|^2 + \| \theta_n^* - \theta_0 \|^2 \right).
    \end{align*}
    Since $\theta_n^*$ is an empirical risk minimizer, the left-hand side is non-negative. Thus,
    $$0 \leq \langle \theta_n^*-\theta_0, \dot{\Psi} \rangle + (P_n-P_0) \alpha_{0,\ell}^{-1} \ell_0'[\dot{\Psi}] + o_p(n^{-1/2}).$$
    
    Similarly, by replacing $\pi_n((1-\epsilon_n) \theta_n^* + \epsilon_n (\theta_0 + \dot{\Psi}))$ with $\pi_n((1-\epsilon_n) \theta_n^* + \epsilon_n (\theta_0 - \dot{\Psi}))$ in \eqref{eq: asymptotic linearity key2}, we derive that
    $$0 \leq -\langle \theta_n^*-\theta_0, \dot{\Psi} \rangle - (P_n-P_0) \alpha_{0,\ell}^{-1} \ell_0'[\dot{\Psi}] + o_p(n^{-1/2}).$$
    Therefore, $| \langle \theta_n^*-\theta_0, \dot{\Psi} \rangle + (P_n-P_0) \alpha_{0,\ell}^{-1} \ell_0'[\dot{\Psi}] |=o_p(n^{-1/2})$. By Conditions~\ref{AdPsi}--\ref{APsiremainder} and Lemma~\ref{Lrate},
    \begin{align*}
    \Psi(\theta_n^*) - \Psi(\theta_0) &= \langle \theta_n^*-\theta_0, \dot{\Psi} \rangle + o_p(n^{-1/2}) \\
    &= -(P_n-P_0) \alpha_{0,\ell}^{-1} \ell_0'[\dot{\Psi}] + o_p(n^{-1/2}).
    \end{align*}
    The asymptotic linearity of $\Psi(\theta_n^*)$ follows. We prove the efficiency in Appendix~\ref{appendix: regular proof}.
\end{proof}

The proof of Theorem~\ref{Tefficiency2} is almost identical.

Nest we present and prove a lemma allows us to interpret Condition~\ref{Aestimation} as an upper bound on the rate of $K$.
\begin{lemma}
    \label{Lestimation}
    Under Conditions~\ref{Aquadraticloss}, \ref{Ainit}, \ref{Aapproxidentity} (\ref{Aapproxidentity2} resp.) and \ref{ALipschitzidentity} (\ref{ALipschitzidentity2} resp.), $\| \pi_n(\theta_0) - \theta_n^\dagger \|=o_p(n^{-1/4})$.
\end{lemma}

\begin{proof}[Proof of Lemma~\ref{Lestimation}]
    By definition of $\theta_n^\dagger$ and Condition~\ref{Aquadraticloss}, we have
    $$\| \theta_n^\dagger - \theta_0 \|^2 \leq C P_0 \{ \ell(\theta_n^\dagger) - \ell(\theta_0) \} \leq C P_0 \{ \ell(\pi_n(\theta_0)) - \ell(\theta_0) \} \leq C \| \pi_n(\theta_0) - \theta_0 \|^2,$$
    the right-hand side of which is $o_p(n^{-1/2})$ by Lemma~\ref{Lrate} (or its corresponding version under Conditions~\ref{ALipschitzidentity2} and \ref{Aapproxidentity2}). Therefore, $\| \theta_n^\dagger - \theta_0 \| = o_p(n^{-1/4})$ and hence $\| \pi_n(\theta_0) - \theta_n^\dagger \| \leq \| \pi_n(\theta_0) - \theta_0 \| + \| \theta_n^\dagger - \theta_0 \| =o_p(n^{-1/4})$.
\end{proof}

We finally prove the efficiency of the data-adaptive series estimator with $K$ selected by CV.

\begin{proof}[Proof of Theorem~\ref{Tcv}]
    By Lemma~\ref{Lrate} and Condition~\ref{Aquadraticloss}, for that existing deterministic $K$, $P_0 \{ \ell(\theta_K^*(\theta_n^0)) - \ell(\theta_0) \} \leq C \| \theta_K^*(\theta_n^0) - \theta_0 \|^2=o_p(n^{-1/2})$. By the oracle inequality for CV in \cite{Vanderlaan2003cv}, $P_0 \{ \ell(\theta_n^\sharp) - \ell(\theta_0) \} = o_p(n^{-1/2})$. By Condition~\ref{Aquadraticloss}, $\| \theta_n^\sharp - \theta_0 \|^2 \leq C P_0 \{ \ell(\theta_n^\sharp) - \ell(\theta_0) \} = o_p(n^{-1/2})$ and hence $\| \theta_n^\sharp - \theta_0 \|=o_p(n^{-1/4})$. So with probability tending to one,
    \begin{align*}
    \| \dot{\psi} \circ \theta_n^0 - \pi_{K^*,\theta_n^0} (\dot{\psi} \circ \theta_n^0) \| &= \| \dot{\psi} \circ \theta_n^0 - \Pi_{K^*,\theta_n^0}(\dot{\psi}) \circ \theta_n^0 \| \\
    &\leq C \| \theta_n^0 - \Pi_{K^*,\theta_n^0}(\id) \circ \theta_n^0 \| & \text{(Condition~\ref{Abadseries})} \\
    &\leq C \| \theta_n^0 - \theta_n^\sharp \| & \text{(definition of the projection operator)} \\
    &\leq C (\| \theta_n^0 - \theta_0 \| + \| \theta_n^\sharp - \theta_0 \|), & \text{(triangle inequality)}
    \end{align*}
    which is $o_p(n^{-1/4})$ by Condition~\ref{Ainit}. Hence,
    \begin{align*}
    \| \dot{\psi} \circ \theta_0 - \pi_{K^*,\theta_n^0}(\dot{\psi} \circ \theta_0) \| &\leq \| \dot{\psi} \circ \theta_0 - \pi_{K^*,\theta_n^0}(\dot{\psi} \circ \theta_n^0) \| \\
    &\leq \| \dot{\psi} \circ \theta_0 - \dot{\psi} \circ \theta_n^0 \| + \| \dot{\psi} \circ \theta_n^0 - \pi_{K^*,\theta_n^0}(\dot{\psi} \circ \theta_n^0) \| \\
    &\leq C \| \theta_n^0 - \theta_0 \| + o_p(n^{-1/4}), & \text{(Condition~\ref{ALipschitzw})}
    \end{align*}
    which is $o_p(n^{-1/4})$ by Condition~\ref{Ainit}.
    
    This bounds the approximation error $\| \dot{\psi} \circ \theta_0 - \pi_{K^*,\theta_n^0}(\dot{\psi} \circ \theta_0) \|$ for $\dot{\psi}$, a result that is similar to Lemma~\ref{Lapproxw} combined with Conditions~\ref{Ainit}~and~\ref{Aapproxw2sufficient}. Similarly to Theorem~\ref{Tefficiency}, along with other conditions, the assumptions in Corollary~2 in \cite{Shen1997} are essentially satisfied and hence an almost identical argument shows that $\Psi(\theta_n^\sharp)$ is an asymptotically linear estimator of $\Psi(\theta_0)$. We prove the efficiency in Appendix~\ref{appendix: regular proof}.
\end{proof}

\subsection{Efficiency} \label{appendix: regular proof}

\begin{proof}[Proof of efficiency of the proposed estimators]
    It is sufficient to show that the influence function of our proposed estimators is the canonical gradient under a nonparametric model. Let $H \in \mathscr{H}$ be fixed. In the rest of this proof, for all small $o$ and big $O$ notations, we let $\delta \rightarrow 0$. The proof is similar to the proof of asymptotic linearity in \cite{Shen1997} except that the estimator of $\theta_0$ and the empirical distribution $P_n$ are replaced by $\theta_{H,\delta}$ and $P_{H,\delta}$ respectively.
    
    Let $\delta'$ satisfy Condition~\ref{Aempiricalprocesslike}. We note that
    \begin{align*}
    P_{H,\delta} \ell(\theta_{H,\delta}) &= P_{H,\delta} \ell(\theta_0) + P_0 [\ell(\theta_{H,\delta}) - \ell(\theta_0)] + (P_{H,\delta} - P_0) [\ell(\theta_{H,\delta}) - \ell(\theta_0)] \\
    &= P_{H,\delta} \ell(\theta_0) + P_0 [\ell(\theta_{H,\delta}) - \ell(\theta_0)] + (P_{H,\delta} - P_0) \ell_0'[\theta_{H,\delta} - \theta_0] \\
    &\quad+ (P_{H,\delta} - P_0) r[\theta_{H,\delta} - \theta_0].
    \end{align*}
    We also note that $(1-\delta')\theta_{H,\delta} + \delta' (\theta_0 \pm \dot{\Psi}) \in \Theta$ if $|\delta|$ is sufficiently small. Then, similarly, by replacing $\theta_{H,\delta}$ with $(1-\delta')\theta_{H,\delta} + \delta' (\theta_0 + \dot{\Psi})$ in the above equation, we have that
    \begin{align}
    \begin{split}
    \label{eq: regularity key2}
    & P_{H,\delta} \ell((1-\delta')\theta_{H,\delta} + \delta' (\theta_0 + \dot{\Psi})) \\
    &= P_{H,\delta} \ell(\theta_0) + P_0 [\ell((1-\delta')\theta_{H,\delta} + \delta' (\theta_0 + \dot{\Psi})) - \ell(\theta_0)] \\
    &\quad+ (P_{H,\delta} - P_0) \ell_0'[(1-\delta')\theta_{H,\delta} + \delta' (\theta_0 + \dot{\Psi}) - \theta_0] \\
    &\quad+ (P_{H,\delta} - P_0) r[(1-\delta')(\theta_{H,\delta} - \theta_0) + \delta' \dot{\Psi}].
    \end{split}
    \end{align}
    Take the difference between the above two equations. By the linearity of $\ell_0'$, we have that
    \begin{align*}
    & P_{H,\delta} \ell((1-\delta')\theta_{H,\delta} + \delta' (\theta_0 + \dot{\Psi})) - P_{H,\delta} \ell(\theta_{H,\delta}) \\
    &= P_0 [\ell((1-\delta')\theta_{H,\delta} + \delta' (\theta_0 + \dot{\Psi})) - \ell(\theta_0)] - P_0 [\ell(\theta_{H,\delta}) - \ell(\theta_0)] \\
    &\quad+ \delta' (P_{H,\delta} - P_0) \ell_0'[\dot{\Psi} - \theta_{H,\delta} + \theta_0] \\
    &\quad+ (P_{H,\delta} - P_0) \{ r[(1-\delta')(\theta_{H,\delta} - \theta_0) + \delta' \dot{\Psi}] - r[\theta_{H,\delta} - \theta_0] \} \\
    &= \frac{\alpha_{0,\ell}}{2} \| (1-\delta')(\theta_{H,\delta} - \theta_0) + \delta' \dot{\Psi} \|^2 - \frac{\alpha_{0,\ell}}{2} \| \theta_{H,\delta} - \theta_0 \|^2 \\
    &\quad+ o\left( \| \theta_{H,\delta} - \theta_0 \|^2 + \| (1-\delta')(\theta_{H,\delta} - \theta_0) + \delta' \dot{\Psi} \|^2 \right) & \text{(Condition~\ref{Aquadraticloss})} \\
    &\quad+ \delta' (P_{H,\delta} - P_0) \ell_0'[\dot{\Psi}] - \delta' (P_{H,\delta} - P_0) \ell_0'[\theta_{H,\delta} - \theta_0] + \delta' o(\delta) & \text{(Condition~\ref{Aempiricalprocesslike})}\\
    &= \delta' \alpha_{0,\ell} \langle \theta_{H,\delta} - \theta_0, \dot{\Psi} \rangle - \delta' \alpha_{0,\ell} \| \theta_{H,\delta} - \theta_0 \|^2 + \delta'^2 \frac{\alpha_{0,\ell}}{2} \| \theta_{H,\delta} - \theta_0 + \dot{\Psi} \|^2 \\
    &\quad+ o\left( \| \theta_{H,\delta} - \theta_0 \|^2 + \| (1-\delta')(\theta_{H,\delta} - \theta_0) + \delta' \dot{\Psi} \|^2 \right) \\
    &\quad+ \delta' (P_{H,\delta} - P_0) \ell_0'[\dot{\Psi}] + \delta' o(\delta) \\
    &\leq \delta' \alpha_{0,\ell} \langle \theta_{H,\delta} - \theta_0, \dot{\Psi} \rangle + \delta'^2 \frac{\alpha_{0,\ell}}{2} \| \theta_{H,\delta} - \theta_0 + \dot{\Psi} \|^2 \\
    &\quad+ o\left( \| \theta_{H,\delta} - \theta_0 \|^2 + \| (1-\delta')(\theta_{H,\delta} - \theta_0) + \delta' \dot{\Psi} \|^2 \right) \\
    &\quad+ \delta' (P_{H,\delta} - P_0) \ell_0'[\dot{\Psi}] + \delta' o(\delta).
    \end{align*}
    Since the left-hand side of the above display is nonnegative, by Condition~\ref{Acloseminimizer}, we have that
    \begin{align*}
    0 &\leq \langle \theta_{H,\delta} - \theta_0, \dot{\Psi} \rangle + \alpha_{0,\ell}^{-1} (P_{H,\delta} - P_0) \ell_0'[\dot{\Psi}] + O(\delta') + o(\delta) \\
    &= \langle \theta_{H,\delta} - \theta_0, \dot{\Psi} \rangle + \alpha_{0,\ell}^{-1} (P_{H,\delta} - P_0) \ell_0'[\dot{\Psi}] + o(\delta).
    \end{align*}
    
    Similarly, by replacing $(1-\delta')\theta_{H,\delta} + \delta' (\theta_0 + \dot{\Psi})$ with $(1-\delta')\theta_{H,\delta} + \delta' (\theta_0 - \dot{\Psi})$ in \eqref{eq: regularity key2}, we show that $0 \leq -\langle \theta_{H,\delta} - \theta_0, \dot{\Psi} \rangle - \alpha_{0,\ell}^{-1} (P_{H,\delta} - P_0) \ell_0'[\dot{\Psi}] + o(\delta)$. Therefore, $|\langle \theta_{H,\delta} + \theta_0, \dot{\Psi} \rangle + \alpha_{0,\ell}^{-1} (P_{H,\delta} - P_0) \ell_0'[\dot{\Psi}]| = o(\delta)$ and
    \begin{align*}
    \Psi(\theta_{H,\delta}) - \Psi(\theta_0) &= \langle \theta_{H,\delta} - \theta_0, \dot{\Psi} \rangle + O(\|\theta_{H,\delta} - \theta_0\|^2) \\
    &= -\alpha_{0,\ell}^{-1} (P_{H,\delta} - P_0) \ell_0'[\dot{\Psi}] + o(\delta) + O(\|\theta_{H,\delta} - \theta_0\|^2) \\
    &= -\alpha_{0,\ell}^{-1} (P_{H,\delta} - P_0) \ell_0'[\dot{\Psi}] + o(\delta). & \text{(Condition~\ref{Acloseminimizer})}
    \end{align*}
    Consequently, $\lim_{\delta \rightarrow 0} [\Psi(\theta_{H,\delta}) - \Psi(\theta_0)]/\delta = P_0 \{ -\alpha_{0,\ell}^{-1} \ell_0'[\dot{\Psi}] \cdot H \}$ and hence the canonical gradient of $\Psi$ under a nonparametric model is $\alpha_{0,\ell}^{-1} \{-\ell_0'[\dot{\Psi}] + P_0 \ell_0'[\dot{\Psi}]\}$. Since the influence functions of our asymptotically linear estimators are equal to this canonical gradient, our proposed estimators are efficient under a nonparametric model.
\end{proof}

\section{Simulation setting details} \label{appendix: simulation}

In all simulations, since $\theta_0(x)=\expect_{P_0}[Z|X=x]$ is the conditional mean function, the loss function was chosen to be the square loss $\ell(\theta): v \mapsto (z-\theta(x))^2$.

\subsection{HAL} \label{appendix: HAL setting}

In the simulation, we generate data from the distribution defined by
$$X \sim \text{N}(0,1), \  \theta_0(x)=\exp\{-(-1+2x+2x^2)/2\}, Z|X=x \sim \text{Exponential}(\text{rate}=1/\theta_0(x)).$$
The sample sizes being considered are 500, 1000, 2000, 5000 and 10000. For each scenario we run 1000 replicates. We chose M.gcv+ to be 3.1 times M.cv.

\subsection{Data-adaptive series}

\subsubsection{Demonstration of Theorem~\ref{Tcv}} \label{appendix: data daptive series setting}

In the simulation, we generate data from the distribution defined by $X \sim \text{Unif}(-1,1), \  Z|X=x \sim \text{N}(\theta_0(x),0.25^2)$ where
\begin{align*}
\theta_0: x &\mapsto I(-1 \leq x < -3/4) + \pi I(-3/4 \leq x < -1/2) + 10 x^2 I(-1/4 \leq x < 1/4) \\
&\quad + \sqrt{2} I(1/4 \leq x < 1/2) + \exp(-1) I(1/2 \leq x < 3/4) + \sqrt[3]{3} I(3/4 \leq x \leq 1),
\end{align*}
When using the trigonometric series, we first shift and scale the initial function range to be $[-1/2,1/2]$ and then use the basis for the interval $[-1,1]$ (i.e. $\sin(j \pi z), \cos(j \pi z)$) in sieve estimation to avoid the poor behavior of trigonometric series near the boundary. We consider sample sizes 500, 1000, 2000, 5000, 10000 and 20000. For each sample size, we run 1000 simulations.

\subsubsection{Violation of Condition~\ref{Abadseries}} \label{appendix: data daptive series rough setting}

In the simulation, we generate data from the distribution defined by $X \sim \text{Unif}(-1,1), \  Z|X=x \sim \text{N}(\theta_0(x),1)$ where $\theta_0: x \mapsto \cos(10 x)$. The estimand is $\Psi(\theta_0)=P_0 (f \circ \theta_0)$ where
\begin{align*}
f: z &\mapsto \left[ \frac{3}{10 \pi} \cos(5 \pi z) - \frac{3}{8} \right] I \left( z < -\frac{1}{2} \right) -\frac{3}{2} z^2 I \left( -\frac{1}{2} \leq z < 0 \right) \\
&\quad\ \, + 3 z^2 I \left( 0 \leq z < \frac{1}{2} \right) + \left[ -\frac{3}{2} \exp(2-4 z) - 3 z + \frac{15}{4} \right] I \left( z \geq \frac{1}{2} \right).
\end{align*}
We consider sample sizes 500, 1000, 2000, 5000, 10000 and 20000; for each sample size, we run 1000 simulations. Our goal is to explore the behavior of the plug-in estimator when $f$, instead of $\theta_0$, is rough, so we use kernel regression \citep{Nadaraya1964} to estimate $\theta_0$ for convenience.

\subsection{Generalized data-adaptive series}

\subsubsection{Demonstration of Theorem~\ref{Tcv2}} \label{appendix: general data daptive series setting}

In the simulation, we generate data from the distribution defined by $X \sim \text{Unif}(-1,1), \  A|X=x \sim \text{Bern}(\expit(-x)), \  Y|A=a,X=x \sim \text{N}(\mu_{0,a}(x),0.25^2)$ where
\begin{align*}
\mu_{00}: x &\mapsto I(-1 \leq x < -3/4) + \pi I(-3/4 \leq x < -1/2) + 10 x^2 I(-1/4 \leq x < 1/4) \\
&\quad + \sqrt{2} I(1/4 \leq x < 1/2) + \exp(-1) I(1/2 \leq x < 3/4) + \sqrt[3]{3} I(3/4 \leq x \leq 1), \\
\mu_{01}: x &\mapsto x^2 I(x < -1/3) + \exp(x) I(-1/3 \leq x < 1/3) + I(x > 1/3)
\end{align*}
The series is the tensor product \citep{Chen2007} of univariate trigonometric series in \ref{appendix: data daptive series setting}. The sample sizes are the same as in \ref{appendix: data daptive series setting}.

\subsubsection{Violation of Condition~\ref{Abadseries2}} \label{appendix: general data daptive series rough setting}

In the simulation, we generate data from the distribution defined by $X \sim \text{Unif}(-1,1), \  A|X=x \sim \text{Bern}(g_0(x)), \  Y|A=a,X=x \sim N(\mu_{0,a}(x),0.25^2)$ where $\mu_{0a}: x \mapsto \exp(-x^2 + 0.8 a x + 0.5 a)$ ($a \in \{0,1\}$) and
\begin{align*}
g_0: x &\mapsto \expit \Bigg\{ \left( -\frac{5}{3} x^3 - \frac{15}{4} x^2 - \frac{5}{3} x - \frac{25}{96} \right) I\left( x \leq -\frac{1}{2} \right) + \left( \frac{5}{6} x^4 + \frac{5}{3} x^3 \right) I\left( -\frac{1}{2} < x \leq 0 \right) \\
&\quad + \frac{5}{3} x^3 I\left( 0 < x \leq \frac{1}{2} \right) + \left( 5x^2 - \frac{15}{4} x + \frac{5}{6} \right) I\left( x > \frac{1}{2} \right) \Bigg\}.
\end{align*}
We consider sample sizes 500, 1000, 2000, 5000, 10000 and 20000; for each sample size, we run 1000 simulations. Our goal is to explore the behavior of the plug-in estimator when $\dot{\Psi}$, instead of $\theta_0$, is rough, so we use kernel regression \citep{Nadaraya1964} to estimate $\theta_0$ for convenience.

\end{appendices}

\section*{Acknowledgements}

This work was partially supported by the National Institutes of Health under award number DP2-LM013340 and R01HL137808. The content is solely the responsibility of the authors and does not necessarily represent the official views of the National Institutes of Health.

\bibliography{references}

\end{document}